\documentclass[11pt]{article}
\usepackage[]{geometry}
\usepackage[mathcal]{euscript}

\usepackage{bm}                     %AMS Packages
\usepackage{amsfonts}
\usepackage{amssymb}
\usepackage{amsmath}
\usepackage{amsthm}
\usepackage{graphicx}               %Graphics
\usepackage{wrapfig}

\usepackage{natbib}                 %Bibliography
\usepackage{colortbl}               %Tables
\usepackage{booktabs}

\usepackage{soul}
\usepackage{todonotes} 

\newcommand{\eps}{\varepsilon}

\usepackage{amssymb,amsmath}

\usepackage[latin1]{inputenc}
\usepackage{graphicx,color}
\newtheorem{definition}{Definition}
\newtheorem{remark}{Remark}
\newtheorem{theorem}{Theorem}
\newtheorem{lemma}{Lemma}
\newtheorem{proposition}{Proposition}
\newtheorem{corollary}{Corollary}

\newcommand{\PROB}{\mathbb{P}}
\newcommand{\R}{\mathbb{R}}

\begin{document}

%\vspace*{2cm}

\begin{center}
\large \bf  Set estimation from reflected Brownian motion \normalsize
\end{center}
\normalsize

\

\begin{center}
 Alejandro Cholaquidis$^{a}$, Ricardo Fraiman$^{b}$,  G\'abor Lugosi$^{c}$ and Beatriz Pateiro-L\'opez$^{d}$\footnote{\textit{Corresponding author: } Departamento de Estad\'{\i}stica e Investigaci\'{o}n Operativa. Facultad de Matem\'aticas. Universidad de Santiago de Compostela. 15782. Santiago de Compostela. Spain. 

\noindent E-mail: beatriz.pateiro@usc.es} \\
  $^{a,b}$ Universidad de la Rep\'ublica, Uruguay\\
  $^{c}$  ICREA \& Universitat Pompeu Fabra, Spain\\
  $^{d}$ Universidad de Santiago de Compostela, Spain\\
\end{center}

%\vspace*{1cm}

\begin{abstract}
We study the problem of estimating a compact set $S\subset \mathbb{R}^d$ from a trajectory of a reflected Brownian motion in $S$
 with reflections on the boundary of $S$. 
We establish consistency and rates of convergence for various estimators
of $S$ and its boundary.
This problem has relevant applications in ecology in estimating the home range of an animal based on tracking data. There are a variety of studies on the habitat of animals that employ the notion of home range. 
%The most common technique is to determine the minimum convex polygon (MCP), which is very restrictive regarding the shape of the habitat. 
%GABOR no creo que las dos siguientes frases sean necesarias. Las he borrado por ahora. Texto modificado un poco.
%Traditional methods for home range estimation assume that the data are independent and identically
%distributed. However, animal movement data cannot be regarded as independent and new home range estimation methods have been proposed to address this issue.
 This paper offers theoretical foundations for a new methodology that, under fairly unrestrictive  shape assumptions, allows one to find flexible regions close to reality. The theoretical findings are illustrated on simulated and real data examples.
\end{abstract}


%\noindent {\bf{AMS 2000 subject classifications:}} Primary 47N30; secondary 60D05, 62G05.

%\vspace*{1cm}

\section{Introduction}\label{sec:intro}

%GABOR: pequeños retoques en el texto
{\textit{Set estimation}}  deals with the problem of approximating, in statistical terms, an unknown compact set $S\subset\mathbb{R}^d$. Most of the related literature assumes that the sampling information is given by  independent observations whose distribution is closely related to the set $S$. When it comes to set estimation methods, the emphasis in the existing literature is mostly
on the geometrical assumptions on $S$ and not on the sampling model. The extent to which a given set estimator efficiently reproduces the unknown set depends heavily on its geometry. Since the early work of \cite{renyi:63, renyi:64},  
%in the convex case, 
significant effort has been made to enlarge the class of sets to estimate, propose efficient estimators, and analyze their asymptotic properties, see \cite{cuevas:09} for a survey. The best known estimator, introduced by \cite{cheva:76}, is simply the union of balls centered at the $n$ sample points of radius $\epsilon_n$.  \cite{dw:80} show that if $\epsilon_n \to 0$ and $n\epsilon_n^d \to \infty$, then the estimate is universally consistent with respect to the measure of the symmetric diference, see (\ref{dmu}) for a formal definition. More precisely, they show that the distance in measure converges to 0 in probability for all absolutely continuous distributions supported on $S$. Regarding estimates with geometric shape restrictions, the works of  \cite{walther:97}, \cite{rodriguez:07} for the class of $r$-convex sets or the more recent work of \cite{chola:14} for the class of {\mbox{$\rho$-cone-convex}} sets show the relevant role of the geometrical assumptions. Recall that a closed convex set can be characterized as the intersection  of all closed half spaces containing the set. An $r$-convex set is characterized as the intersection of complements of open balls of radius $r$  that do not intersect the set.  This condition implies the outside $r$-rolling condition  according to which, informally,  at every point of the boundary one can place an open ball included in the complement of the set. In both cases, letting $r\to\infty$, one obtains the convex sets. However, the class of $r$-convex sets is much more flexible, allowing also holes and smooth inlets in the set, see Figure \ref{fig:rconvex}. In terms of set estimation, if a set is assumed to be $r$-convex, then it can be efficiently approximated from the $r$-convex hull of a random sample of points taken into it, see \cite{rodriguez:07}. 

\begin{figure}[!h]%
	\centering
	\includegraphics[scale=0.4]{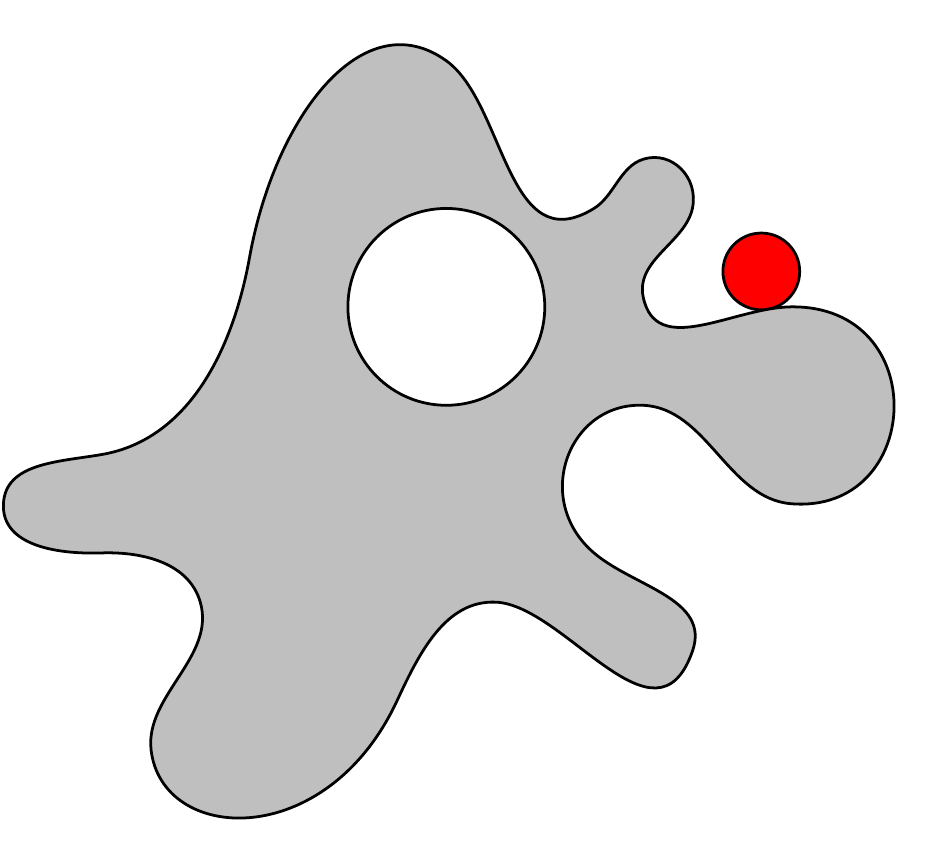} %Vista satelite
		\hspace{1cm}
			\includegraphics[scale=1.2]{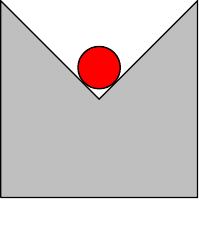}
	\caption{On the left, an example of an $r$-convex set with inlets and holes. The value of $r$ corresponds to the radius of the largest ball that can roll outside the set. 
%GABOR: no sé si se verán los colores. Mejor no hablar de ellos.
%(represented in red). 
The set on the right is neither convex  nor $r$-convex for any value of $r$.
	}%
	\label{fig:rconvex}%
\end{figure}

The R-package \texttt{alphahull}, described in \cite{pateiro:10}, provides a practical implementation
of the $r$-convex hull estimator for the i.i.d.\ case in dimension $d = 2$. In this work, we are concerned with another aspect of set estimation that has been less studied. We are interested in the problem of estimating an unknown set $S$ from a trajectory of a stochastic process that lives in the set. Up to our knowledge there are no results in the set estimation literature for this framework. 
We work with the model of  {\it{reflected Brownian motion}}. While it is an admittedly simplistic model for many applications, it offers a general, rigorous, and well understood framework. Possible extensions and generalizations are discussed in Section \ref{sec:diffusion} below.

%GABOR varios cambios en el texto aquí
In Section \ref{setup}, we discuss conditions for the existence of the reflected Brownian motion and its stationary uniform distribution and establish connections between these probabilistic conditions and geometric constraints on its support. In Section \ref{cons} we prove consistency of several estimates of $S$ based on a trajectory of a reflected Brownian motion (RBM). We also describe geometric conditions that ensure consistency.
In particular, we introduce and study the behaviour of two estimates, the $r$-convex hull of a trajectory and the so-called {\mbox{RBM-sausage}}. 
These estimates may be considered as analogues of the $r$-convex hull of a random sample and the \cite{dw:80} estimate in the i.i.d.\ case, respectively. The study of properties of the Brownian sausage goes back to 1933 (see \cite{kolmo:33}). This estimator is closely related to kernel density estimation (KDE) methods. Indeed, the estimate of \cite{dw:80} is nothing but the set where a kernel density estimator using the uniform kernel is positive. On the other hand, the $r$-convex hull provides a quite flexible family of estimators, with good asymptotic properties in the i.i.d.\ case.
In Section \ref{rates} we obtain upper bounds for the rates of convergence.
%GABOR: frase borrada
%	 The proposed estimators can be computed in practice, by adapting the implementation in the i.i.d. case. 

In Section \ref{sec:diffusion} we generalize the stochastic model generating the 
observed trajectories. In particular, we consider the general class of {\it{reflected
diffusions}}, a class of stochastic processes that include reflected Brownian
motion as a special case. This class allows one to deal with processes with
non-uniform stationary distribution on the support, an important aspect of some of the 
applications in home-range estimation.
We show that both estimators considered for reflected Brownian motion 
are well defined and consistent for any kind of reflected diffusion under mild conditions. However, 
%translating the probabilistic conditions to sufficient geometric conditions on the support and 
obtaining rates of convergence remain 
a challenge in the general case.
In Section \ref{coments} we describe the simulation and real-data studies that
illustrate the behavior of the set estimation methods described in the paper.
The code used in the paper is available in a new release of the R package \texttt{alphahull 2.0}.  
Before introducing the formal framework, we
discuss the application of the proposed methodology in home-range estimation from animal tracking data, through real data examples.

%{\textit{home-range estimation.}}
%\textit{Home-range estimation from animal tracking data}\label{dataHR}

Home-range estimation is a principal concern in animal ecology. Home range was first defined by \cite{burt:43} as ``the area traversed by the individual in its normal activities of food gathering, mating, and caring for young''. Since this first definition, the concept of home range has evolved, giving rise to a considerable amount of literature on the subject (reviews are given, for instance, by \cite{worton:87} and \cite{powell:00}).
%been refined and a considerable amount of literature has been published on the subject, see \cite{okubo:01} for a review.
The home range of an animal is usually estimated from a set of locations collected over a period of time. A first approach was to estimate the animal's home range by means of the convex hull of the observed location points (\cite{hayne:49}). This  ``minimum-convex-polygon'' method has well-known shortcomings. A major drawback is that the estimated home range can include areas of land which are never used. This overestimation can be reduced with the use of more flexible estimators such as the $\alpha$-hull, see \cite{burgman:03}.
Other home-range estimation methods describe the animal's home range by the so-called utilization distribution (density function that describes the probability of finding the animal at a particular location).  Since their introduction by \cite{worton:89} in the context of home-range data, methods based on kernel density estimation 
%(\cite{silverman:86}) 
have been widely used for estimating the utility distribution. We refer to \cite{seaman:96} for an evaluation of the performance of these methods. More recently, \cite{getz:04} and \cite{getz:07} proposed the  nonparametric kernel method ``local convex hull'' that estimates the utilization distribution from the union of local nearest-neighbour convex hulls.

These methods of home-range estimation generally treat the recorded locations as independent observations. However, the advances in animal tracking technology (VHF radio transmission, Argos system, GPS, etc.) have allowed one to almost continuously record the movements of animals. In this context, the independence of observations cannot be assumed and new 
mathematical models
%home-range estimation methods 
are needed. 
Modelling the movement of an animal in its home range as a continuous stochastic process provides a more realistic framework in which tracking data can be analyzed.  Existing stochastic models for describing animal movement can be found in \cite{okubo:01}, \cite{preisler:04} and references therein.
Other relevant references include \cite{bor:08}, \cite{fry:08}, \cite{pat:08} and \cite{tang:10}. 
Some recent results consider more involved statistical problems in the home-range setup which are not covered by our proposal, which only deals with the  ``densely sampled data".
\cite{flem:15} consider the problem of home-range estimation when only a short trajectory is observed and propose a method called  AKDE (autocorrelated KDE)
that takes into account the autocorrelations to provide a bandwidth (typically much larger) to be used in KDE and predict future animal movement. \cite{buchi:12} consider the case where the trajectories 
are only observed  at a low sampling rate and analyse animal movement using the Brownian Bridge movement model (BBMM; \cite{horne:07}), which selects a Brownian Bridge trajectory between each pair of nearest 
points in time in the low sampling rate original trajectory. On the contrary, the densely sampled data corresponds to a high sampling rate. As mentioned in \cite{kie:10} ``the closer locations are in  time, as obtained using GPS technology, the closer locations are in space, and kernel estimators can estimate utilization distribution well without the need for Brownian Bridge''. As we show it in this paper, our proposal works well in this setup. \cite{ben:11} proposes to use movement-based kernel density estimation (MKDE) to estimate the utilization distribution using the circular bivariate Gaussian kernel, where the bandwidth varies for each data-point.

As an alternative approach, mechanistic models (see \cite{moor:06}, \cite{potts:14}) incorporate interaction behavior to characterise home ranges. For example, \cite{potts:13} (see also \cite{giu:11}) consider a model where several individuals interact in the home range.  ``Animals are modeled to move at random but constrained to roam within areas that do not contain scent of cospecifics". The scent persists for a limited amount of time. Otherwise the individuals perform a nearest neighbor random walk (NNRW) or a ballistic walk (BW).

\

%GABOR frase eliminada (no entiendo qué queremos decir)
%The methodology proposed in this paper, in the framework of set estimation from a trajectory
%of a stochastic process, could suggest alternative home-range estimation methods. 

Because of the fairly unrestrictive nature of the shape assumptions, the methods based on the  $r$-convex hull of a trajectory and the RBM-sausage can identify hard boundaries in the home range. Moreover, the estimation is based on a trajectory of a stochastic proccess, more in line with the recent literature on home-range estimation methods. It can be argued that the reflected Brownian motion is a simplified model for animal movement. One of the main limitations of this model
is that the stationary distribution is necessarily uniform over the domain, which may not be
a realistic assumption for animal movement as it does not allow one to contemplate
the notion of ``core area'' (the area where the animals spend most of the time).
A more general related model that incorporates non-uniform stationary distributions is
reflected Brownian motion with drift (\cite{kr:14}, \cite{hw:87}).
As an alternative, in Section \ref{sec:diffusion} we discuss the general model of reflected
diffusions, a class of stochastic processes that include reflected Brownian
motion as a special case. 
%\sout{To keep the notation and discussion simple, we focus on reflected Brownian motion. %\color{blue}{The theory developed may provide insight for more complex models.}}
%\textcolor{red}{Before starting with more technical aspects, we consider a real data %example as well as a simulated one.}
%\subsection{Application to home range estimation from animal tracking %data}\label{dataHR}

\

As mentioned before, advances in tracking technology %(VHF transmission, GPS and other satellite systems) 
have provided researchers the opportunity to obtain large amount of tracking data from a large variety of species with a high temporal resolution. Movebank  %(\cite{wikelski:14}) 
is an online database that gives open access to animal movement data collected by researchers. 
As an illustration, we have considered data from the ``Dunn Ranch Bison Tracking Project''. Over the last years, the Nature Conservancy in Missouri ({http://www.nature.org/}) has been working on the restoration of the Dunn Ranch Prairie, located in northwest Missouri. A herd of bison was introduced onto the ranch in order to restore the prairie ecosystem. The bison roam across a large fenced area. In Figure \ref{fig:dunn} we show the movements of two bison with $n=9635$ (left), $n=19380$ (right) recorded positions. In red, we represent the boundary of one of the proposed estimators, the $r$-convex hull estimator of the trajectory, for $r=0.005$.

\begin{figure}[!h]%
	\centering
	\includegraphics[scale=0.4]{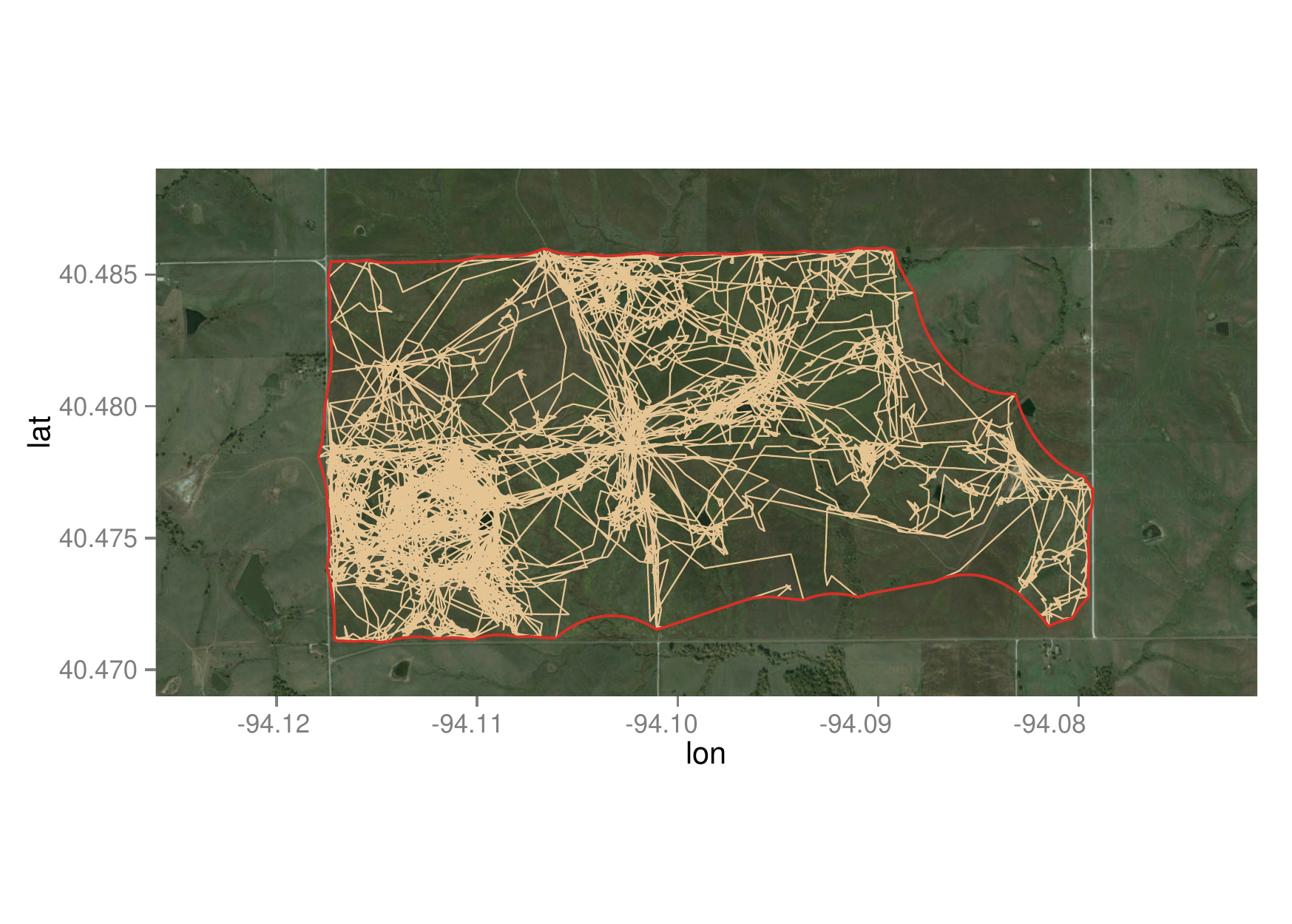} %Vista satelite
	\includegraphics[scale=0.4]{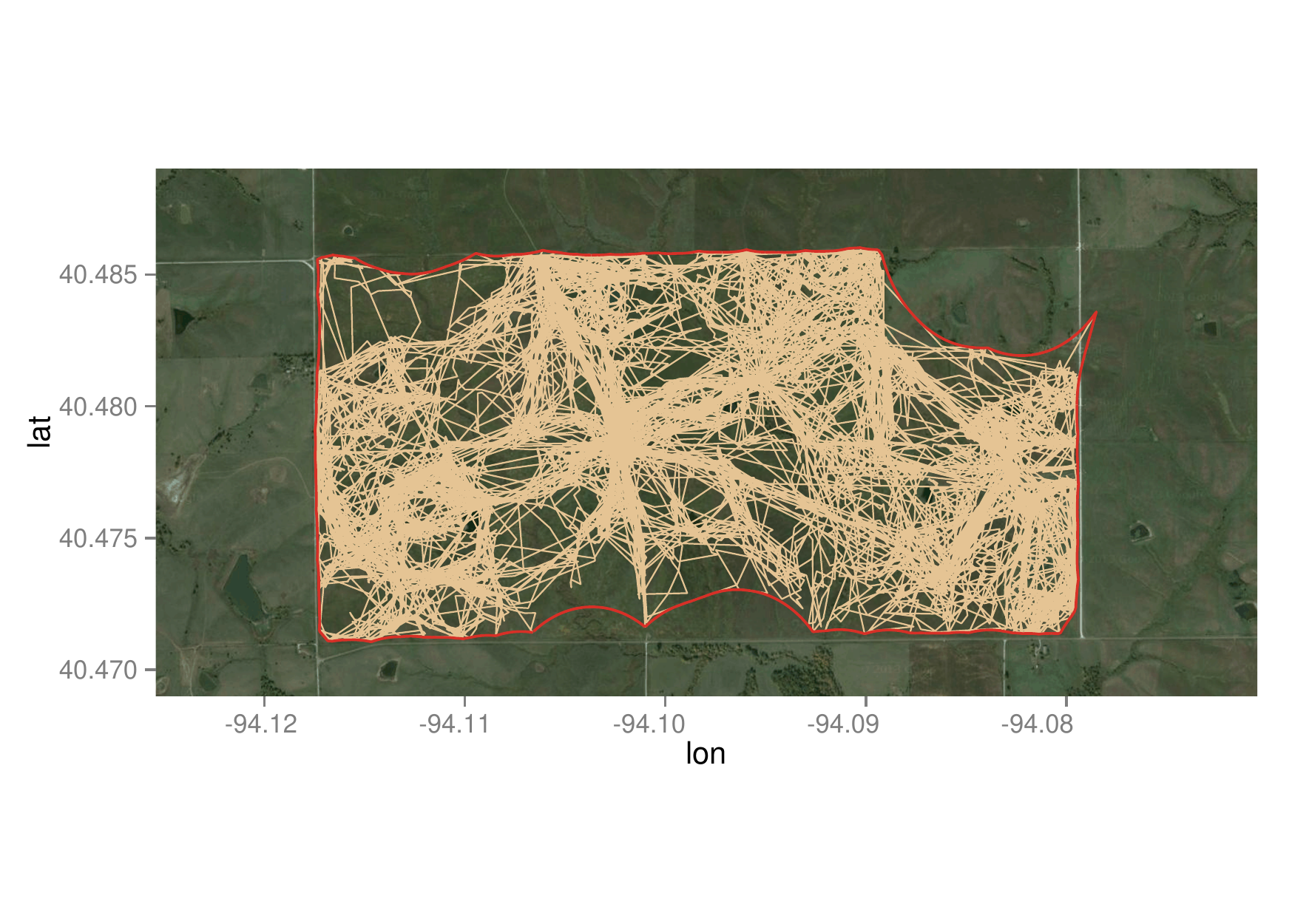}
	%\vspace{-1.5cm}
	
	%\includegraphics[scale=0.35]{dunn_sat19380.pdf}
	
	%\vspace{-1.5cm}
	
	%\includegraphics[scale=0.5]{dunn_sat30210.pdf}
	%\caption{Movements of three bison in the Dunn Ranch Prairie with $n=9635$ (top), $n=19380$ (middle) and $n=30210$ (bottom) recorded positions. In red, boundary of the $r$-convex hull estimator for $r=0.005$.
	\caption{Movements of two bison in the Dunn Ranch Prairie with $n=9635$ (left), $n=19380$ (right) recorded positions. In red, boundary of the $r$-convex hull estimator for $r=0.005$.
	}%
	\label{fig:dunn}%
\end{figure}

For the bison with $n=9635$ recorded positions, we have computed the continuous version of the Devroye-Wise estimator (the reflected Brownian sausage), for different values of the smoothing parameter $\epsilon_T$, see Figure \ref{fig:dunn2}. A detailed discussion of these estimators is given in the next sections. In Section \ref{coments}, we analyse the behaviour of the $r$-convex hull estimator with respect to (i) how much the estimated home range differs when we observe the real movement to pass from time 0 to $T$, 0 to $2T$, etc. (short trajectories) and (ii) a variation on the discretization time (low sampling rate).

\begin{figure}[!h]%
	\centering
	\includegraphics[scale=0.4]{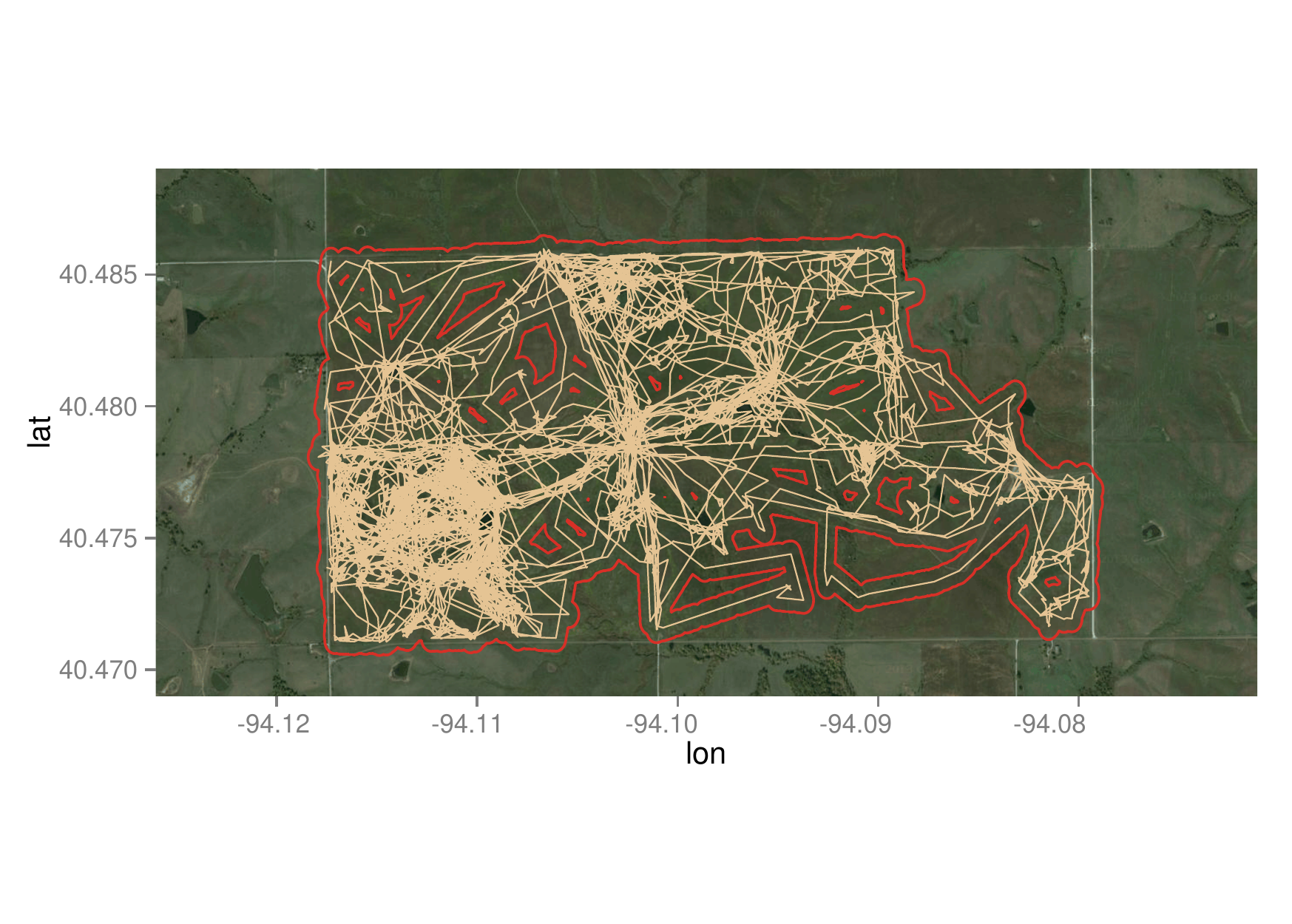} %Vista satelite
	\includegraphics[scale=0.4]{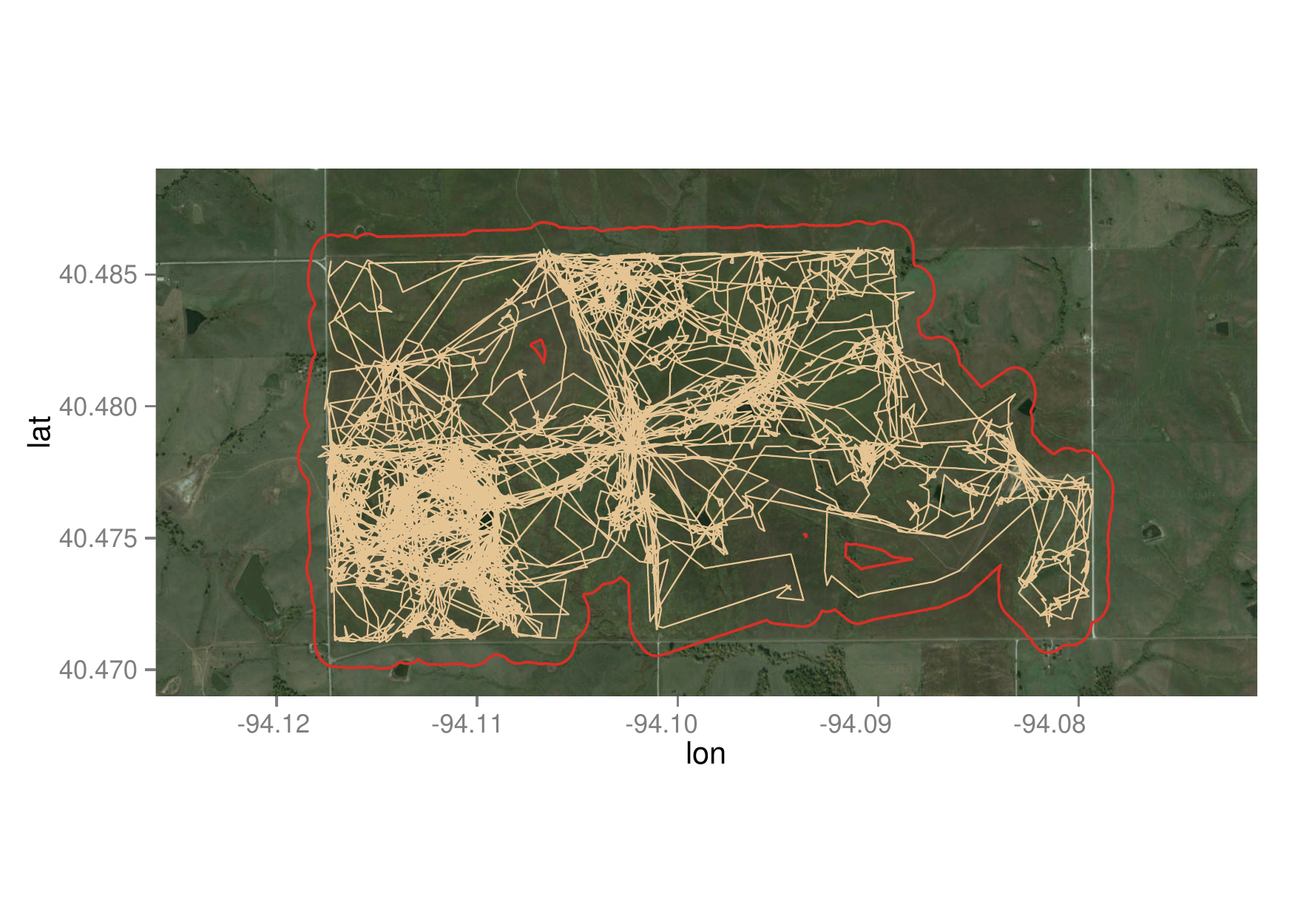}
	%\vspace{-1.5cm}
	
	%\includegraphics[scale=0.5]{dunn_sat9635_DW0_00075.pdf}
	
	%\vspace{-1.5cm}
	
	%\includegraphics[scale=0.5]{dunn_sat9635_DW0_001.pdf}
	%\caption{Movements of one bison in the Dunn Ranch Prairie with $n=9635$ recorded positions. In red, boundary of the reflected Brownian sausage $D_T$ for $\epsilon_T=0.0005$ (top), $\epsilon_T=0.00075$ (middle) and $\epsilon_T=0.001$ (bottom).
	%}%
	\caption{Movements of one bison in the Dunn Ranch Prairie with $n=9635$ recorded positions. In red, boundary of the reflected Brownian sausage $D_T$ for $\epsilon_T=0.0005$ (left)  and $\epsilon_T=0.001$ (bottom).
	}%
	\label{fig:dunn2}%
\end{figure}
\

\section{Setup}\label{setup}

In this section we establish conditions for the existence of the reflected Brownian motion and its stationary uniform distribution, and study the connections between these conditions and some geometric constraints on its support. In particular, we analyze how the regularity required on the support is naturally related to  rolling-type properties, which are usual assumptions in set estimation. We first introduce some notation and basic definitions used throughout the manuscript.\\

{\textit{Notation and basic definitions.}} \\
Given a set $S$, we denote by $\partial S$, $\textnormal{int}(S)$, and $\overline{S}$ the boundary, interior, and closure of $S$, respectively. We denote by $\left\langle\cdot,\cdot \right\rangle$  the usual inner product in $\mathbb{R}^d$ and by $\left\|\cdot\right\|$ the Euclidean norm.

Let $\mathcal{B}(x,\epsilon)$ denote the closed ball of radius $\epsilon$ centred at $x$. 
The open ball is denoted by $\mathring{\mathcal{B}}(x,r)$.  Given a bounded set $A\subset \mathbb{R}^d$ and $\epsilon>0$, $B(A,\epsilon)$ denotes the parallel set $B(A,\epsilon)=\{x\in \mathbb{R}^d:\ d(x,A)\leq \epsilon\}$, where $d(x,A)=\inf\{\|x-a\|:\ a\in A\}$. 

Given $x\in \mathbb{R}^d$, a unit vector $\xi\in \mathbb{R}^d$, and $\rho\in (0,\pi/2]$, $C(x,\xi,\rho)$ denotes  the convex cone with vertex $x$, angle $\rho$, and orientation $\xi$, defined by 
\[C(x,\xi,\rho)=\big\{y\in\mathbb{R}^d:\ \left\langle y-x,\xi\right\rangle\geq \left\|y-x\right\|\cos\rho \big\}~.\]
The closed compact cone of radius $h$ is $C_h(x,\xi,\rho)=C(x,\xi,\rho)\cap \mathcal{B}(x,h)$.

The performance of a set estimator %$S_T$ 
is usually evaluated through the Hausdorff distance (\ref{dh}) and the distance in measure (\ref{dmu}) given below. 
The distance in measure takes the mass of the symmetric difference into account while the Hausdorff distance measures the difference of the shapes.

Let $A,C\subset \mathbb{R}^d$ be non-empty and compact. The Hausdorff distance between $A$ and $C$ is defined as 
\begin{equation}\label{dh}
d_H(A,C)=\max\Big\{\max_{a\in A}d(a,C), \ \max_{c\in C}d(c,A)\Big\}.
\end{equation}

If $\mu$ is a Borel measure, the distance in measure between  $A$ and $C$ is defined as 
\begin{equation}\label{dmu}
d_\mu(A,C)=\mu(A\triangle C),
\end{equation}
where $\triangle$ denotes symmetric difference.

\medskip

{\textit{The reflected Brownian motion.}} 

Let $D$ be a domain in $\mathbb{R}^d$ (that is, a connected and open set) with closure $\overline{D}$ and boundary  $\partial D$. We are concerned with the problem of existence and uniqueness of solution for a reflected stochastic differential equation on the domain $D$. 
This problem has been discussed by \cite{tanaka:79} when $D$ is a convex domain and by \cite{lions:84} and \cite{saisho:87} when $D$ is a general domain satisfying certain regularity conditions. 
In particular, \cite{saisho:87} proved that if $D$ satisfies the {\it{uniform exterior sphere}} condition and the {\it{uniform interior cone}} condition (formalized in Definitions \ref{def:A} and \ref{intcone2} below), 
%\cite{saisho:87}  proves that if $D$ is a domain in $\mathbb{R}^d$ satisfying both conditions, 
then there exists a unique strong solution of the Skorohod stochastic differential equation
\begin{equation}\label{sde}
X_t=X_0+B_t+\int_{0}^t\eta(X_s)dL_s,\ t\geq 0,
\end{equation}
where $B_t$ is a $d$-dimensional Brownian motion, $\eta$ denotes the inward unit vector on the boundary $\partial D$ and $L$ is a continuous nondecreasing process with $L_0=0$ and
\[L_t=\int_0^t\mathbb{I}_{\{X_s\in\partial D\}}dL_s.\]
Roughly speaking, the process behaves in the interior of the set like an ordinary Brownian motion and reflects at the boundary.

\medskip

Following the notation by \cite{saisho:87}, for $x\in\partial D$, let
\[\mathcal{N}_x=\bigcup_{r>0}\mathcal{N}_{x,r},\]
\begin{equation*} \label{unifsph}
\mathcal{N}_{x,r}=\big\{\eta\in \mathbb{R}^d:\ \left\|\eta\right\|=1,\ \mathring{\mathcal{B}}(x-r\eta,r)\cap D=\emptyset\big\}.
\end{equation*}

\begin{definition}\label{def:A} 
The domain $D$ satisfies the  uniform exterior sphere condition if there exists a constant $r_0>0$ such that, for any $x\in\partial D$,
\begin{equation} \label{cond1}
\mathcal{N}_{x}=\mathcal{N}_{x,r_0}\neq\emptyset. 
\end{equation}
\end{definition}

For each $x\in \partial D$, the family of sets $\mathcal{N}_{x,r}$ is decreasing with respect to $r$. Condition \eqref{cond1} means that there exists $r_0$ such that for all $x\in \partial D$, taking $r\leq r_0$ does not add any new direction in $\mathcal{N}_{x,r}$. It is not easy to characterize geometrically this condition. However we prove below that the family of sets that satisfy \eqref{cond1} is between two well-known classes: it is contained in the class of $r_0$-convex sets (see \cite{bram12}), and contains the class of sets that satisfies the outside and inside $r_0$-rolling condition (stated in the proof of Proposition \ref{rollthennontrap}).

\begin{definition}\label{intcone2}The domain $D$ satisfies the uniform interior cone condition if there exist $\delta>0$ and $\rho\in (0,\pi/2]$  such that for any $x\in\partial D$ there exists a unit vector $l_x$ with
\begin{equation} \label{cond2}
C(y,l_x,\rho)\cap \mathring{\mathcal{B}}(x,\delta)\subset \overline{D}, \quad \forall y\in \mathring{\mathcal{B}}(x,\delta)\cap \partial D.
\end{equation}

\end{definition}

In \cite{bram12} it is shown that this condition is equivalent to the domain being Lipschitz (see for instance Definition 9 in \cite{bram12}). From a geometric point of view this condition is related to the cone-convexity property introduced in \cite{chola:14}: a set $S\subset {\mathbb R}^d$ is  $\rho$-cone-convex, for some $\rho\in(0,\pi/2]$, if there exists $\delta>0$ such that for all $x\in\partial S$ there is an open cone, denoted $\mathring{C}(x,l_x,\rho)$, such that $\mathring{C}(x,l_x,\rho)\cap \mathring{\mathcal{B}}(x,\delta)\subset S^c$. Condition \eqref{cond2} is stronger than the $\rho$-cone-convexity property in the sense that it requires that the same direction $l_x$ works for a neighborhood of $x$.\\

\textit{The trap condition.} 

Apart from the uniform exterior sphere condition and the uniform interior cone condition, another important notion is that of a {\it{non-trap}} domain.
 Let $\mathcal{B}\subset D$ and consider the first hitting time of $\mathcal{B}$ by $X$, $T_\mathcal{B} = \inf\{t>0 : X_t \in \mathcal{B}\}$.

\begin{definition}\label{def:trap}
As defined by \cite{burdzy:06}, we say that $D$ is a trap domain 
for the stochastic process $X_t$
if 
there exists a closed ball  $\mathcal{B}\subset D$ with non-zero radius
such that
\begin{equation} \label{trapcond}
\sup_{x\in D}\mathbb{E}^xT_\mathcal{B}=\infty,
\end{equation}
where $\mathbb{E}^x$ denotes the expectation of the distribution of 
$X_t$ starting from $x$.  Otherwise $D$ is called a non--trap domain. 
\end{definition}

It is proved in Lemma 3.2 in \cite{burdzy:06} that if $X_t$ is a reflected Brownian motion in

 a connected open set $D$ with finite volume and $\mathcal{B}_1$ and $\mathcal{B}_2$ are closed non--degenerate balls in $D$, then
$\sup_{x\in D}\mathbb{E}^xT_{\mathcal{B}_1}<\infty$ if and only if $\sup_{x\in D}\mathbb{E}^xT_{\mathcal{B}_2}<\infty$.

\

The non--trap condition is related to the uniform ergodicity of reflected Brownian motion in $D$. Indeed, this is shown in the following proposition given in  \cite{burdzy:06}. Condition (iii) will be used in Section \ref{rates} to obtain the rates of convergence of the proposed estimators.

\begin{proposition} (\cite{burdzy:06}, Prop. 1.2)\label{prop:burdzy} Let $D \subset \mathbb R^d$ be a connected open set with finite volume, and denote by $\Pi_D$ the uniform probability measure in $D$. 
Let $X_t$ be the reflected Brownian motion in $D$.
Then the following are equivalent.
\begin{itemize}
\item [(i)] $D$ is a non-trap domain for $X_t$;
\item [(ii)] $\lim_{t \to \infty} \sup_{x \in D} \Vert \mathbb P^x(X_t \in \cdot) - \Pi_D \Vert _{TV} = 0$;
\item [(iii)] There are positive constants $\alpha$ and $\beta$ such that 
$$
\sup_{x \in D} \Vert \mathbb P^x (X_t \in \cdot ) - \Pi_D\Vert_{TV} \leq \beta e^{-\alpha t};
$$
\end{itemize}
where $\Vert  \mu \Vert_{TV}$ stands for the total variation norm of the measure $\mu$.
\end{proposition}

\medskip

Consider a non-empty compact set $S\subset \mathbb{R}^d$ with connected interior. We show in Proposition \ref{Bthennontrap} below that, if $S=\overline{\textnormal{int}(S)}$ and $\textnormal{int}(S)$ satisfies the uniform interior cone condition, then $\textnormal{int}(S)$  is non-trap. To formalize the argument of the proof, we need the following definitions.
%From the statistical point of view of set estimation, it is natural to assume that the domain is non-trap.
%clear that if $T_\mathcal{B}$ is not finite, then there will be no consistent estimate of the set $S$ based on a trajectory of the reflected Brownian motion on $\textnormal{int}(S)$. 
%We show that if $S=\overline{\textnormal{int}(S)}$ and $\textnormal{int}(S)$ satisfies the uniform interior cone condition given in Definition \ref{intcone2}, then $\textnormal{int}(S)$ is non-trap. A first geometric condition on a domain $D \subset \mathbb R^d$ in order to be non--trap follows from Corollary 2.9 and Proposition 1.4 in \cite{burdzy:06}, where it is shown that every John domain  (see Definition \ref{johndomain} below) with finite volume is a non--trap domain. 

%A first geometric condition on a domain $D \subset \mathbb R^d$ in order to be non--trap 

\begin{definition} \label{intcone1} A bounded domain $D$ satisfies the interior cone condition if there exists $\rho\in (0,\pi/2]$ and $h>0$ such that for all $x\in D$ there exists a unit vector $\xi_x$, such that $C_h(x,\xi_x,\rho)\subset D$. \end{definition}

%\begin{remark} \label{rem1} By Lemma 7 in \cite{hajlasz:01}, if a domain satisfies the interior cone condition, then it is a John domain and then, if it has finite volume, it is \textit{non-trap}.
%\end{remark}

\begin{proposition} \label{Bthennontrap} Consider a non-empty compact set $S\subset \mathbb{R}^d$ with connected interior. Suppose that $S=\overline{\textnormal{int}(S)}$ and $\textnormal{int}(S)$ satisfies the uniform interior cone condition. Then $\textnormal{int}(S)$ is a \textit{non-trap} domain for the reflected Brownian motion $X_t$ in $\textnormal{int}(S)$.%satisfies the interior cone condition.
\begin{proof} 
First, we prove that, under the stated assumptions, $\textnormal{int}(S)$ satisfies the interior cone condition given  in Definition \ref{intcone1}. Reasoning by contradiction, if $\textnormal{int}(S)$ does not satisfy the interior cone condition, there exists $x_n\in \textnormal{int}(S)$ , and two sequences $\rho_n>0$ and $h_n>0$ with $\rho_n,h_n\rightarrow 0$ such that
\begin{equation} \label{eq1prop} 
\forall \xi \ \text{ with } \left\|\xi\right\|=1, \quad  C_{h_n}(x_n,\xi,\rho_n)\cap \textnormal{int}(S)^c\neq \emptyset.
\end{equation}
As $S$ is compact, $\{x_n\}$ has a subsequential limit $x\in S$. We may assume, (taking a subsequence if necessary) that $x_n\rightarrow x$. As $h_n\rightarrow 0$ and (\ref{eq1prop}) holds, we have that $x\in \partial S$. Since $\textnormal{int}(S)$ satisfies the uniform interior cone condition, there  exist $\delta>0$, $\rho \in (0,\pi/2]$, and a unit vector $l_x$ such that (\ref{cond2}) holds. Let us take $n$ large enough such that $h_n<\delta/4$. Taking $\xi=-l_x$ in (\ref{eq1prop}), there exists $y_n\in \partial S$ and $y_n\in C_{h_n}(x_n,-l_x,\rho)\cap \textnormal{int}(S)^c$. As $h_n<\delta/4$ we have that $y_n\in \mathcal{B}(x,\delta/2)$. Again by the uniform interior cone condition, we have that $C(y_n,l_x,\rho)\cap \mathring{\mathcal{B}}(x,\delta)\subset S$ and $x_n\in C(y_n,l_x,\rho)$ (see Figure \ref{propimg}). We consider two cases. First, if $x_n \in \textnormal{int}(C(y_n,l_x,\rho))$, then $C_{h_n}(x_n,l_x,\rho)\subset  \textnormal{int}(C(y_n,l_x,\rho))\subset \textnormal{int}(S)$ which contradicts (\ref{eq1prop}). Now, if $x_n \in \partial C(y_n,l_x,\rho)$ (as in Figure \ref{propimg}) then we can take a unit vector $\nu$ such that $C_{h_n}(x_n,\nu,\rho)\subset C(y_n,l_x,\rho)$ and, as $x_n\in \textnormal{int}(S)$, $C_{h_n}(x_n,\nu,\rho) \subset \textnormal{int}(S)$ which, again,  contradicts (\ref{eq1prop}).
\begin{figure}[!h]%
\centering
\includegraphics[scale=2]{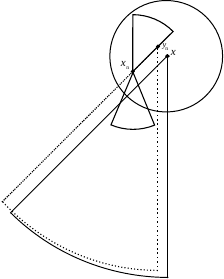}
\caption{In solid lines the ball $\mathcal{B}(x,\delta/4)$ and the cones $C_\delta(x,l_x,\rho)$, $C_{h_n}(x_n,-l_x,\rho)$ and $C_{h_n}(x_n,\nu,\rho)$. In dotted lines the cone $C_\delta(y_n,l_x,\rho)$.  }%
\label{propimg}%
\end{figure}

Now, by Lemma 7 in \cite{hajlasz:01}, if a domain satisfies the interior cone condition, then it is a John domain (i.e., there is a constant $C\geq  1$ and a distinguished point $x_0\in D$ such that each point $x\in D$ can be joined to $x_0$ by a curve $\gamma:[0,1]\rightarrow D$ such that $\gamma(0)=x$, $\gamma(1)=x_0$ and $d(\gamma(t),\partial D)\geq C^{-1}\|x-\gamma(t)\|$). Corollary 2.9 and Proposition 1.4 in \cite{burdzy:06} show that every John domain with finite volume is a non--trap domain. 
\end{proof}
\end{proposition}

\medskip

\textit{Rolling-type conditions.} 
\
We have shown that if $S\subset \mathbb{R}^d$ is a non-empty compact set with with connected interior and $S=\overline{\textnormal{int}(S)}$ such that $\textnormal{int}(S)$  satisfies both the  uniform exterior sphere condition and the uniform interior cone condition, then $\textnormal{int}(S)$ is non-trap and there exists a unique strong solution of the Skorohod stochastic differential equation (\ref{sde}). Next we analyze how the required regularity on $\partial S$ is related to three well known rolling-type properties (positive reach, $r$-convexity and outside rolling condition). These rolling-type properties have been used in set estimation as shape restrictions that cover large families of sets (much larger than the family of convex sets). We refer to \cite{cuevas:12} for more details.

\medskip

Following the notation in \cite{federer:59}, let $\text{Unp}(S)$ be the set of points $x\in \mathbb{R}^d$ having a unique projection on $S$, denoted by $\xi_S(x)$. That is, for $x\in \text{Unp}(S)$, $\xi_S(x)$ is the unique point that achieves the minimum of $\|x-y\|$ for $y\in S$. We write $\delta_S(x)=\inf\{\|x-y\|:y\in S\}$.
\begin{definition} \label{reach} For $x\in S$, let {\it{reach}}$(S,x)=\sup\{r>0:\mathring{\mathcal{B}}(x,r)\subset Unp(S)\big\}$. The reach of $S$ is defined by
$$reach(S)=\inf\big\{reach(S,x):x\in S\big\},$$
and $S$ is said to be of positive reach if $reach(S)>0$.
\end{definition}

\begin{definition} \label{rconvexity} A set $S\subset \mathbb{R}^d$ is said to be $r$-convex, for $r>0$, if 
$$S=C_r(S),$$
where
$$C_r(S)=\bigcap_{\big\{ \mathring{\mathcal{B}}(x,r):\ \mathring{\mathcal{B}}(x,r)\cap S=\emptyset\big\}} \Big(\mathring{\mathcal{B}}(x,r)\Big)^c$$
is the $r$-convex hull of $S$.
\end{definition}

\begin{definition} \label{rolling} Let $S\subset \mathbb{R}^d$ be a closed set. A ball of radius $r$ is said to roll freely in $S$ if for each boundary point $s\in \partial S$ there exists some $x\in S$ such that $s\in \mathcal{B}(x,r)\subset S$. The set $S$ is said to satisfy the outside $r$-rolling condition if a ball of radius $r$ rolls freely in $\overline{S^c}$.
\end{definition}

%\begin{remark} \label{reach-rconvex}
The relationship between positive reach, $r$-convexity, and outside rolling condition is analyzed in \cite{cuevas:12}. It is proved that the class of sets with reach $r$ is included in the class of $r$--convex sets, which is included in the class of sets satisfying the outside $r$--rolling condition. The following result by  \cite{bram14} shows the relation of the uniform exterior sphere condition and the uniform interior cone condition to the rolling-type conditions.

\begin{lemma}(\cite{bram14}, Lemma A.3)  \label{ABimplyreach} Let $D$ be a bounded domain that satisfies the uniform exterior sphere condition and the uniform interior cone
condition. Then $\overline{D}$ has positive reach.
\end{lemma}

\begin{remark}
The relationship between the constant $r_0$ of the uniform exterior sphere condition, the angle $\rho$ of the uniform interior cone condition, and the value of $\textnormal{reach}(\overline{D})$ is discussed in  \cite{bram14}.
\end{remark}

\begin{remark}
A direct consequence of Lemma  \ref{ABimplyreach} and Propositions 1 and 2 in \cite{cuevas:12} is that if $S\subset \mathbb{R}^d$ is a non-empty compact set such that $\textnormal{int}(S)$  satisfies both the  uniform exterior sphere condition and the uniform interior cone condition, then $S$ satisfies the outside rolling condition. In fact, we only need to assume that $\textnormal{int}(S)$  satisfies the uniform exterior sphere condition with radius $r_0$ to prove that $S$ satisfies the outside $r_0$--rolling condition.
%he outside rolling condition can be proved assuming only that  $\textnormal{int}(S)$  satisfies the uniform exterior sphere condition. Moreover, if $r_0$ is the uniform exterior ball radius, then $S$ satisfies the outside $r_0$--rolling condition.  
Note that, if $r_0$ is the uniform exterior ball radius, for any $s\in \partial S$,  $\mathcal{N}_{s,r_0}\neq \emptyset$. Then there exists $\eta:=\eta(s)\in \mathcal{N}_{s,r_0}$ such that $\mathring{\mathcal{B}}(s-r_0\eta,r_0)\cap \textnormal{int}(S)=\emptyset$. That is, $s\in \mathcal{B}(s-r_0\eta,r_0)\subset \overline{S^c}$ (a ball of radius $r_0$ rolls freely in $\overline{S^c}$).
The converse implication is not true. %Let us consider the set $S=\partial \Big(\mathcal{B}\big((0,0),1\big)\cup \mathcal{B}\big((1,0),1\big)\Big)$, see Figure \ref{rout}. We have that $x=(1/2,\sqrt{3}/2)\in \partial S$ and $y=(1/2,-\sqrt{3}/2)\in \partial S$. Note that $\mathcal{N}_y=C_{1}(y,(0,-1),2\pi/3)$ but $\mathcal{N}_{y,1}= \{v_1,v_2\}$ being $v_1=(-1/2,-\sqrt{3}/2)$ and $v_2=(1/2,-\sqrt{3}/2)$. 
Consider the set $S=R\setminus\textnormal{int}(A)$, with $R=\left[-1.5,2.5\right]\times\left[-1.5,1.5\right]$ and $A=\mathcal{B}\big((0,0),1\big)\cup \mathcal{B}\big((1,0),1\big)$, see Figure \ref{rout}. It is clear that a ball of radius $r_0=1$ rolls freely in $\overline{S^c}$. We have that $x=(1/2,-\sqrt{3}/2)\in \partial S$. Note that $\mathcal{N}_x=\{\eta: \|\eta\|=1, \langle \eta,\xi\rangle\geq \cos(\pi/6)\}$, being $\xi=(0,-1)$, but $\mathcal{N}_{x,1}= \{v_1,v_2\}$ being $v_1=(-1/2,-\sqrt{3}/2)$ and $v_2=(1/2,-\sqrt{3}/2)$. Therefore, $\textnormal{int}(S)$  does not satisfy the uniform exterior sphere condition with $r_0=1$.
\end{remark}

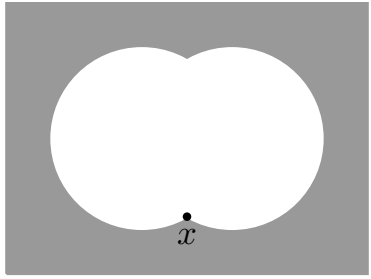
\begin{figure}[ht]
\begin{center}
\scalebox{1.2}{\begin{tikzpicture}[thick]
\draw[draw=gray!80!white,fill=gray!80!white] (-1.5,-1.5) -- (-1.5,1.5)-- (2.5,1.5)-- (2.5,-1.5) --(-1.5,-1.5);
\draw[fill=white,draw=white] (0,0) circle (1cm);
\draw[fill=white,draw=white] (1,0) circle (1cm);
\draw[color=black,fill=black] (0.5,-0.866) circle (0.2ex);
\node[below] at (0.5,-0.866) {$x$};
\end{tikzpicture}}
\caption{%The set $S=\partial \Big(\mathcal{B}\big((0,0),1\big)\cup \mathcal{B}\big((1,0),1\big)\Big)$ satisfies the $r$ outer rolling condition but condition A is not fulfilled with $r_0=r$. 
In gray, $S=R\setminus\textnormal{int}(A)$, with $R=\left[-1.5,2.5\right]\times\left[-1.5,1.5\right]$ and $A=\mathcal{B}\big((0,0),1\big)\cup \mathcal{B}\big((1,0),1\big)$. The set $S$ satisfies the outside $r_0$--rolling condition for $r_0=1$ but $\textnormal{int}(S)$  does not satisfy the uniform exterior sphere condition for $r_0=1$.}
\label{rout}%
\end{center}
\end{figure}

%\begin{figure}[!h]%
%\centering
%\includegraphics[scale=1.5]{rout.pdf}
%\caption{The set $S=\partial \Big(\mathcal{B}\big((0,0),1\big)\cup \mathcal{B}\big((1,0),1\big)\Big)$ satisfies the $r$ outer rolling condition but condition A is not fulfilled with $r_0=r$}
%\label{rout}%
%\end{figure}

\begin{proposition} \label{rollthennontrap} Let $S\subset \mathbb{R}^d$ be a non-empty compact set with connected interior satisfying $S=\overline{\textnormal{int}(S)}$. Suppose that
for some $r_0>0$, a ball of radius $r_0>0$ rolls freely in $S$ and $\overline{S^c}$. Then the reflected Brownian motion in $\textnormal{int}(S)$ exists and $\textnormal{int}(S)$ is non-trap.

\begin{proof} By Proposition \ref{Bthennontrap} it is enough to prove that $\textnormal{int}(S)$ satisfies the uniform exterior sphere condition and the uniform interior cone condition. We prove first that the uniform interior cone condition holds. By Theorem 1 of \cite{walther:99}, $\partial S$ is a $(d-1)$-dimensional $C^1$ submanifold, and the outward unit vector $\eta(x)$ in a point $x\in \partial S$ is Lipschitz. Then there exists $\delta_0>0$ such that for all $x,y\in \partial S$, if $\|x-y\|<\delta_0$, the angle between $\eta(x)$ and $\eta(y)$ is smaller than $\pi/6$. As a ball of radius $r_0$ rolls freely in $\overline{S^c}$ for every point $x\in \partial S$, $C_{r_0}(x,-\eta(x),\pi/3)\subset \mathcal{B}(x-\eta(x)r_0,r_0)\subset S$. If $y\in \mathring{\mathcal{B}}(x,\delta_0)\cap \partial S$, $C_r(y,-\eta(x),\pi/6)\subset C_r(y,-\eta(y), \pi/3)\subset \overline{S^c}$. Then, the uniform interior cone condition is satisfied with $\delta=\delta_0$ and $\rho=\pi/6$. To prove that the uniform exterior sphere condition is satisfied, observe that by Lemma 2.3 in \cite{pateiro:09}, for all $x\in \partial S$, $\mathcal{N}_{x,r_0}= \eta(x)$ and then $\mathcal{N}_{x,r_0}\neq \emptyset$. It remains to prove that $\mathcal{N}_x=\mathcal{N}_{x,r_0}$. Since $\mathcal{B}(x-r_0\eta(x),r_0)\subset S$ it follows that if for some $r$, $\mathcal{N}_{x,r}\neq \emptyset$ then $\mathcal{N}_{x,r}=\eta(x)$, and then $\mathcal{N}_x=\eta(x)$.
\end{proof}
\end{proposition}

%This uniform exterior sphere condition seems to be related shape restrictions that generalize that of convexity. 
%\begin{definition}
%Given $r>0$,
%a set $S\subset\mathbb{R}^d$ is said to fulfil the (outside)
%$r$-rolling condition if for all $x\in \partial S$ there is a
%closed ball with radius $r$, $B_x$, such that $x\in B_x$ and
%$\mbox{int}(B_x) \cap S=\emptyset$.
%\end{definition}

\section{On the consistency of the estimators}\label{cons}

Denote by $W_T=\{X_t:0\leq t\leq T\}$ the trajectory of the reflected Brownian motion
in a domain $S$, up to time $T$. We prove the consistency of several estimates of $S$
based on observing $W_T$, both in terms of the Hausdorff distance and distance in 
measure. Theorem \ref{thcons} establishes the consistency of any estimate $S_T$ containing the trajectory,
under the condition that the estimate is contained within $S$ or within $B(S,\eps_T)$, for $\eps_T\rightarrow 0$. We may apply the result to two estimates in particular, the $r$-convex hull of a trajectory, $C_r(W_T)$, and the so-called RBM-sausage
$D_T=B(W_T, \epsilon_T)=\{x\in \R^d: \exists t\in [0,T] \ \text{such that} \ \|x-X_t\|\le \epsilon_T\}$.
% (analogous in our context to the $r$-convex hull of a random sample of points and the
%\cite{dw:80} estimate in the i.i.d. case, respectively).

\begin{theorem}\label{thcons} Let $S\subset \mathbb{R}^d$ be a compact set satisfying $S=\overline{\textnormal{int}(S)}$. Suppose that $\textnormal{int}(S)$ is connected and satisfies conditions (\ref{cond1}), (\ref{cond2}).
\begin{itemize}
	\item[i)] If $S_T$ is any set such that $W_T\subset S_T\subset S$ a.s.,
	then, with probability one
	\begin{equation*}
	d_H(S_T,S)\rightarrow 0\quad \text{as }t\rightarrow +\infty.
	\end{equation*}
	\item[ii)] The same result holds under the weaker condition $W_T\subset S_T\subset B(S,\eps_T)$ for some sequence $\eps_T\rightarrow 0$.
	\item[iii)] In particular, $C_r(W_T)$ and $D_T$ are consistent estimates since they satisfy the condition $W_T\subset S_T\subset B(S,\eps_T)$.
%	\item[iv)] \textcolor{orange}{If we assume that the domain is non-trap instead of the geometric conditions (\ref{cond1}) and (\ref{cond2}), items $i)$, $ii)$ and $iii)$ hold for any stochastic process defined on the domain $D$. }
\end{itemize}

\begin{proof} $i)$ As $S_T\subset S$, $\forall \eps>0$ we have $S_T\subset B(S_T,\eps)\subset B(S,\eps)$ and we only need to prove that for all $\eps>0$, $S\subset B(S_T,\eps)$ a.s. for $T$ large enough. Reasoning by contradiction, suppose that there exists $\eps>0$ such that $\forall T>0$ there exists $x_T\in S$ but $x_T\notin B(S_T,\eps)$. Then $\mathcal{B}(x_T,\eps)\cap S_T=\emptyset$. As $S$ is compact, there exists $x\in S$ and a sequence $T_n \to \infty$ such that $x_{T_n}\rightarrow x$ as $n\rightarrow +\infty$. Clearly $\mathring{\mathcal{B}}(x,\eps)\cap S_T=\emptyset$ for all $T$. As $S=\overline{\textnormal{int}(S)}$, we can take $z\in \mathcal{B}(x,\eps/2)$ and $0<\delta<\eps/2$ such that $\mathcal{B}(z,\delta)\subset \mathring{\mathcal{B}}(x,\eps)\subset S$. Since $S$ satisfies the non-trap condition (\ref{trapcond}), with probability one, there exists a time $T_1$ such that $X_{T_1}\in \mathcal{B}(z,\delta)$, a contradiction.  The proof of ii) is similar. 
%\textcolor{red}{To prove $iv)$ observe that in the proof of $i)$ we have only used 
%(\ref{cond1}), (\ref{cond2}) to guarantee the non-trap condition.}
\end{proof}
\end{theorem}

The next theorem establishes the consistency of the $r$-convex hull of the trajectory $W_t$ in terms of the distance in 
measure.
In order to use this estimate, one must know a value $r$ for which $S$ is guaranteed to be
$r$-convex. Below we discuss how one may avoid this condition by choosing $r$ in a
data-dependent manner (see Remark \ref{choiceofr}). 

\begin{theorem} 
\label{thm:rhullcons}
Let $S\subset \mathbb{R}^d$ be a compact set satisfying $S=\overline{\textnormal{int}(S)}$. Suppose that $\textnormal{int}(S)$ is connected and satisfies conditions (\ref{cond1}), (\ref{cond2}). Let $r\in (0,r_0)$ where $r_0$ is the constant given in (\ref{cond1}).
Then, with probability one,
%\begin{itemize}
%\item[i)] 
\begin{equation*}
d_\mu(C_r(W_T),S)\rightarrow 0\quad \text{as }T\rightarrow +\infty,
\end{equation*}
where $\mu$ denotes the Lebesgue measure in $\mathbb{R}^d$.
%\item[ii)] 
%\todo{Mismo problema como antes.}
%\textcolor{red}{If we assume that the domain is non-trap instead of the geometric conditions (\ref{cond1}) and (\ref{cond2}), item $i)$ hold for any reflected diffusion.}
%\end{itemize}

\begin{proof} Observe that $\mu(\partial S)=0$. By Theorem 3 in \cite{cuevas:12} we have that $d_H\big(\partial C_r(W_T),\partial S\big)\rightarrow 0$, which implies that $d_\mu(C_r(W_T),S)\rightarrow 0\quad \text{as }T\rightarrow +\infty.$
\end{proof} 
\end{theorem}

In the next result we prove that the surface area of the $r$-convex hull of the trajectory $W_T$ is a consistent estimate of the surface area of any $r$-convex set. 
To make the statement precise, we need the following definitions and a key result by \cite{federer:59}.

\begin{definition} A set $S$ satisfies the property of interior local connectivity if there exists $\alpha_0>0$ such that for all $\alpha\leq \alpha_0$ and for all $x\in S$, $int(\mathcal{B}(x,\alpha)\cap S)$ is a non-empty connected set.
\end{definition}

\begin{definition} \label{minkcont} The outer Minkowski content of $S\in \mathbb{R}^d$  is given by,
$$L_0(\partial S)=\lim_{\epsilon\rightarrow 0} \frac{\mu\big(B(S,\epsilon)\setminus S\big)}{\epsilon},$$
provided that the limit exists and it is finite.
\end{definition}

\begin{theorem}(\cite{federer:59}, Th. 5.6)  Let $S \subset \mathbb R^d$ be a compact set with $reach(S)>0$. Let $K$ be a Borel subset of $\mathbb R^d$. Then there exist unique Radon measures $\Phi_0(S,.), \ldots, \Phi_d(S, .)$ over $\mathbb R^d$ such that for $0 \leq \epsilon < reach(S)$,
$$
\mu(B(S, \epsilon) \cap \{x: \zeta_S(x) \in K\}) = \sum_{i=0}^d \epsilon^{d-i} b_{d-i} 
\Phi_{i}(S,K),
$$
where $b_0=1$, $\zeta_S(x)$ is the unique projection of $x$ on $S$ and, for $j\geq 1$, $b_j$ is the $j$\nobreakdash-dimensional measure of a unit ball in $\mathbb R^j$. 
\end{theorem}
The measures $\Phi_j$ are the curvature measures associated with $S$, and in particular, $L_0(\partial S) = \Phi_{d-1}(S, \partial S)$. 

In what follows we show that the Minkowski content of the $r$-hull of a discretization of the trajectory $W_T$ provides a consistent estimate of the surface area of the set $S$ in the two-dimensional case. Note that the result does not imply that $L_0(\partial C_r(W_T))\rightarrow L_0(\partial S)$. 
%However, in practice we will always have only a discretization of the trajectory. 
To prove such a statement one would need that
$C_r(W_T)$ satisfies the  property of interior local connectivity, which is unclear to us.

\begin{theorem} \label{teomink} Let $S\subset \mathbb{R}^2$ be a compact set satisfying $S=\overline{\textnormal{int}(S)}$. Suppose that $\textnormal{int}(S)$ is connected and satisfies conditions (\ref{cond1}), (\ref{cond2}) and the interior local connectivity property. 
Let $A_{N(T)}\subset W_T$ be a sequence of subsets of finite cardinality such that $d_H(A_{N(T)},S)\rightarrow 0$ a.s.\ as $T\rightarrow \infty$. Then
\begin{equation} \label{minkconv}
L_0\big(\partial C_r(A_{N(T)})\big)\rightarrow L_0(\partial S) \quad a.s.
\end{equation}
\end{theorem}
\begin{proof}
The proof makes use of the following extension of Lemma 1 in \cite{cuevas:12}
\begin{lemma} \label{lemmink} Let $S\subset \mathbb{R}^2$ be a compact $r$-convex set satisfying $S=\overline{\textnormal{int}(S)}$ and the property of interior local connectivity.
If $A_N\subset S$ is a sequence of sets with finite cardinality such that $d_H(A_{N},S)\rightarrow 0$, then there exists $r_0>0$ such that
$\text{reach}\big(C_r(A_N)\setminus I(C_r(A_N))\big)>r_0$ for all $N$, where $I(C_r(A_N))$ is the set of isolated points of $C_r(A_N)$ (i.e., $I(C_r(A_N))=\{x\in C_r(A_N): B(x,\eta)\cap C_r(A_N)=x \text{ for some }\eta>0\}$).
\end{lemma}
The proof of Lemma \ref{lemmink} follows the same lines of that of Lemma 1 in \cite{cuevas:12} with minor changes, so we omit it.\

\

Now, from Lemma \ref{lemmink} we have that 
both the set
$\tilde{S}_{N(T)}= C_r\big(A_{N(T)}\big)\setminus I\big(C_r(A_{N(T)})\big)$ and $S$ have positive reach. Since we also have that $d_H(\tilde{S}_{N(T)},S)\rightarrow 0$ a.s., and
the assumptions of Theorem 5.9 in \cite{federer:59}
are satisfied (see also Remark 4.14 in that paper) the proof will be complete. Indeed, Theorem 5.9 in \cite{federer:59} establishes that the curvature measures are continuous with respect to $d_H$ (see Remark 5.10 in \cite{federer:59}). In particular we obtain that 
$\Phi_{d-1}(\tilde{S}_{N(T)},K)\rightarrow \Phi_{d-1}(S,K)$ for any closed ball $K$ such that $S\subset K$. Using Remark 5.8 in \cite{federer:59} and $\tilde{S}_{N(T)} \subset S$ we get that $\Phi_{d-1}(\tilde{S}_{N(T)},K)=\Phi_{d-1}(\tilde{S}_{N(T)},K\cap \partial \tilde{S}_{N(T)})=\Phi_{d-1}(\tilde{S}_{N(T)},\partial \tilde{S}_{N(T)})$ and also $\Phi_{d-1}(S,K)=\Phi_{d-1}(S,\partial S)$. The proof of (\ref{minkconv}) is concluded by noting that $L_0(\partial \tilde{S}_{N(T)})=\Phi_{d-1}(\tilde{S}_{N(T)},\partial \tilde{S}_{N(T)})$ and $L_0(\partial S)=\Phi_{d-1}(S,\partial S)$.
\end{proof}

%%%%%%%%%%%%%%%%%%%%%%%%%%%%%%%

\section{Rates of convergence}\label{rates}

In this section we establish upper bounds for the rates of convergence of the set estimates
discussed in the previous section (the $r$-convex hull of a trajectory and the RBM-sausage), both for the expected Hausdorff distance and for the expected distance in measure. 

Let $S\subset \mathbb{R}^d$ be a compact set satisfying $S=\overline{\textnormal{int}(S)}$. Suppose that $D=\textnormal{int}(S)$ is connected. Recall that, if $D$ is non-trap then, by Proposition \ref{prop:burdzy} (iii), there exist constants $\alpha, \beta >0$ such that
$$
\sup_{x \in D} \Vert \mathbb P^x (X_t \in \cdot ) - \Pi_D\Vert_{TV}
  \leq \beta e^{-\alpha t}~.
$$
These constants and the volume of the domain $S$ play a key role in
the following
estimate of the rates of convergence of the set estimate.

\begin{theorem}
\label{thm:rates}
Let $S\subset \mathbb{R}^d$ be a compact set such that $S=\overline{\textnormal{int}(S)}$ and $\textnormal{int}(S)$ is a non-trap domain. Let $\{X_t:t\ge 0\}$ be a reflected
Brownian motion in $S$. Let $S_T$ be any measurable set containing
the trajectory $W_T$ such that $S_T\subset S$.
Then for all $T>0$ and for $\epsilon< 2(2\beta \mu(S)/v_0)^{1/d}$
(where $v_0=\mu(\mathcal{B}(0,1))$
is the volume of the unit
ball in $\R^d$),
\[
\PROB\{d_H(S_T,S) > \epsilon\}
\le
\frac{(\epsilon/4)^{-d} \mu(S)}{v_0}
\exp\left(- T \frac{(\epsilon/2)^dv_0\alpha}{2\mu(S)\log \frac{2\beta \mu(S)}{v_0(\epsilon/2)^d} }\right)~.
\]
\end{theorem}

\begin{proof}
Define $\delta= v_0(\epsilon/2)^d/(2\mu(S))$. Let
\begin{equation} \label{eqn}
  n= \left\lfloor \frac{T}{\frac{1}{\alpha}\log \frac{\beta}{\delta}}
     \right\rfloor,
\end{equation}
and define $t_i=\frac{i}{\alpha}\log \frac{\beta}{\delta}$ for $i=1,\ldots,n$.
Note that the condition for $\epsilon$ guarantees that $\beta/\delta>1$.
(Roughly speaking, $t_1,\ldots,t_n$ divide the interval $[0,T]$ in $n$ intervals of
length $\frac{1}{\alpha}\log \frac{\beta}{\delta}$.)

Denote the $\epsilon$-inner parallel set of $S$ 
by
\begin{equation}\label{inner}
  S^{(\epsilon)}=\{x\in S: \mathcal{B}(x,\epsilon)\subset S\}~.
\end{equation}
Then
\begin{eqnarray*}
\PROB\{d_H(S_T,S) > \epsilon\} & \le &
\PROB\{\exists x\in S^{(\epsilon)}: \forall t\in [0,T]: X_t\notin \mathcal{B}(x,\epsilon) \} \\
& \le & \PROB\{\exists x\in S^{(\epsilon)}: \forall i\in \{1,2,\ldots,n\}: X_{t_i}\notin \mathcal{B}(x,\epsilon) \}~.
\end{eqnarray*}
Let $x_1,\ldots,x_N\in S^{(\epsilon)}$ be such that
\[
S^{(\epsilon)} \subset \mathcal{B}(x_1,\epsilon/2) \cup \cdots \cup \mathcal{B}(x_N,\epsilon/2),
\]
and $N$ is the smallest positive integer such that such covering of $S^{(\epsilon)}$
is possible. $N=N(\epsilon/2)$ is called the $\epsilon/2$-covering number of
$S^{(\epsilon)}$. It is easy to see (and well known) that
$N\le \mu(S)/\mu(\mathcal{B}(0,\epsilon/4)) = (\epsilon/4)^{-d} \mu(S)/v_0$.

If for some $x\in S$ we have $X_{t_i}\notin \mathcal{B}(x,\epsilon)$ for all
$i=1,\ldots,n$, then there exists a  $j\in \{1,\ldots,N\}$ such that
 $X_{t_i}\notin \mathcal{B}(x_j,\epsilon/2)$ for all
$i=1,\ldots,n$. Thus, continuing the chain of inequalities above,
\begin{eqnarray*}
\PROB\{d_H(S_T,S) > \epsilon\}
& \le & \PROB\{\exists j\in \{1,\ldots,N\}: \forall i\in \{1,2,\ldots,n\}: X_{t_i}\notin \mathcal{B}(x_j,\epsilon/2) \}  \\
& \le &
N\sup_{x\in S^{(\epsilon)}} \PROB\{\forall i\in \{1,2,\ldots,n\}: X_{t_i}\notin \mathcal{B}(x,\epsilon/2) \}
\end{eqnarray*}
Next, we estimate the probability on the right-hand side. For all $x\in S$,
\begin{eqnarray} \label{eq2}
\lefteqn{
\PROB\{\forall i\in \{1,2,\ldots,n\}: X_{t_i}\notin \mathcal{B}(x,\epsilon/2) \}  }  \\
& = &
\PROB\{X_{t_n}\notin \mathcal{B}(x,\epsilon/2)| \forall i\in \{1,2,\ldots,n-1\}: X_{t_i}\notin \mathcal{B}(x,\epsilon/2) \}  \nonumber  \\
& & \qquad \times
\PROB\{\forall i\in \{1,2,\ldots,n-1\}: X_{t_i}\notin \mathcal{B}(x,\epsilon/2) \}\nonumber \\
& = &
\PROB\{X_{t_n}\notin \mathcal{B}(x,\epsilon/2)| X_{t_{n-1}}\notin \mathcal{B}(x,\epsilon/2) \}
\times
\PROB\{\forall i\in \{1,2,\ldots,n-1\}: X_{t_i}\notin \mathcal{B}(x,\epsilon/2) \}\nonumber \\
& & \text{(since $X_t$ is a Markov process)}\nonumber
\end{eqnarray}
Now, by Proposition \ref{prop:burdzy} (iii),
\begin{eqnarray*}
\PROB\{X_{t_n}\notin \mathcal{B}(x,\epsilon/2)| X_{t_{n-1}}\notin \mathcal{B}(x,\epsilon/2) \}
& \le & 1-\frac{(\epsilon/2)^dv_0}{\mu(S)} +\beta e^{-\alpha (t_n-t_{n-1})}
\\
& = & 1-\frac{(\epsilon/2)^dv_0}{\mu(S)} + \delta  \\
& = & 1-\frac{(\epsilon/2)^dv_0}{2\mu(S)}
\end{eqnarray*}
by the definition of $\delta$.
Hence, we have
\[
\PROB\{\forall i\in \{1,2,\ldots,n\}: X_{t_i}\notin \mathcal{B}(x,\epsilon/2) \}
\le \left(1-\frac{(\epsilon/2)^dv_0}{2\mu(S)}  \right)
\PROB\{\forall i\in \{1,2,\ldots,n-1\}: X_{t_i}\notin \mathcal{B}(x,\epsilon/2) \}
\]
and by iterating the argument,
\[
\PROB\{\forall i\in \{1,2,\ldots,n\}: X_{t_i}\notin \mathcal{B}(x,\epsilon/2) \}
\le \left(1-\frac{(\epsilon/2)^dv_0}{2\mu(S)}  \right)^n
\le \exp\left(- n \frac{(\epsilon/2)^dv_0}{2\mu(S)}\right)~.
\]
Summarizing, and substituting the value of $n$ and $\delta$, we have
\[
\PROB\{d_H(S_T,S) > \epsilon\}
\le
\frac{(\epsilon/4)^{-d} \mu(S)}{v_0}
\exp\left(- T \frac{(\epsilon/2)^dv_0\alpha}{2\mu(S)\log \frac{2\beta \mu(S)}{v_0(\epsilon/2)^d} }\right)~.
\]
\end{proof}

Theorem \ref{thm:rates} implies that, ignoring logarithmic factors,
$d_H(S_T,S)$ is roughly of the order of $(T\alpha/\mu(S))^{-1/d}$.
More precisely, we have the following:

\begin{corollary} \label{hausrate}
Let $S\subset \mathbb{R}^d$ be a compact set such that $S=\overline{\textnormal{int}(S)}$ and $\textnormal{int}(S)$ is a non-trap domain.
Let $\{X_t:t\ge 0\}$ be a reflected
Brownian motion in $S$. Let $S_T$ be any measurable set containing
the trajectory $W_T$ such that $S_T\subset S$.
Then 
\[
d_H(S_T,S) =o\Big(\Big(\frac{\log(T)^2}{T}\Big)^{1/d}\Big) \quad a.s.
\]
\begin{proof} Since $d_H(S_T,S)$ is non-increasing it suffices to prove that 
$$d_H(S_n,S) =o\Big(\Big(\frac{(\log(n))^2}{n}\Big)^{1/d}\Big) \quad a.s.$$
We show that, for all $\epsilon>0$, $\sum_{n=1}^{\infty} \mathbb{P}(a_n d_H(S_n,S)>\eps)<\infty$ where $a_n= \Big(\frac{K\alpha}{3}\frac{n}{(\log(n))^2}\Big)^{1/d}$ and then
apply the Borel-Cantelli lemma. Observe that, if we denote $K=\nu_0\epsilon^d/(2^{d+1}\mu(S))$, then
\begin{align*}
\mathbb{P}(a_nd_H(S_n,S)>\epsilon)\leq&\  a_n^dK^{-1}2^{d-1}\exp\Big(-n K\alpha\frac{1}{a_n^d\log(K^{-1}\beta a_n^d)}\Big)\\
                                   =&\  K^{-1}2^{d-1}\exp\Big(\log(a_n^d)-n K\alpha\frac{1}{a_n^d\log(K^{-1}\beta a_n^d)}\Big).
                                   \end{align*}
Then $\sum_{n=1}^{\infty} \mathbb{P}(a_n d_H(S_n,S)>\eps)<\infty$ follows from the fact that 
$$\lim_{n\rightarrow+\infty} \frac{1}{\log(n)}\Big[\log(a_n^d)-n K\alpha\frac{1}{a_n^d\log(K^{-1}\beta a_n^d)}\Big]=-2.$$
\end{proof}

\end{corollary}

We may now apply the previous results to analyze the RBM-sausage  $D_T=B(W_T, \epsilon_T)$ for some appropriately chosen decreasing function 
$\epsilon_T$.

\begin{theorem}
\label{thm:rbmsausage}
Let $S\subset \mathbb{R}^d$ be a compact set such that $S=\overline{\textnormal{int}(S)}$ and $\textnormal{int}(S)$ is a non-trap domain. Let $\{X_t:t\ge 0\}$ be a reflected
Brownian motion in $S$. Assume that the surface area $L_0(\partial S)$ exists. 
Let $D_T=B(W_T, \epsilon_T)$, with
$\epsilon_T\sim \Big(\frac{\log(T)^2}{T}\Big)^{1/d}$. Then
$d_\mu(D_T,S)= \mathcal{O}\Big(\Big(\frac{\log(T)^2}{T}\Big)^{1/d}\Big)$ a.s.
\begin{proof} Observe that $\mu(D_T,S)=\mu(D_T\setminus S)+\mu(S\setminus D_T)$. 
By the existence of $L_0(\partial S)$ we have
$\mu(D_T\setminus S)=\mathcal{O}(\epsilon_T)$. Since $W_T\subset S$ we may apply Corollary \ref{hausrate} to conclude that
$d_H(W_T,S)= o(\epsilon_T)$ a.s., and then, for $T$ large enough, $S\subset D_T$, a.s., so $\mu(S\setminus D_T)=0$.
\end{proof}
\end{theorem}

\begin{remark} As it is typically the case in non-parametric estimation, the choice of the smoothing parameter is a crucial point. Since in practice the data of the trajectory are always discretized, we suggest two different approaches to select $\eps_T$, built from the discretization of the trajectory.   
		\begin{itemize}
		\item If we assume that the set $S$ is connected, we may use the analogue of the proposal of \cite{ba:00} for the i.i.d case. More precisely, the smoothing parameter $\eps_T$ is chosen as:
		$$\overline{\epsilon_T}=\inf\left\{\epsilon>0:D_T \text{ is connected} \right\}.$$
		An easy-to-implement algorithm is also proposed in \cite{ba:00}, reminiscent of the minimal spanning tree.
		\item For the general case, we suggest the following procedure. After discretization, split the sample
 at random 
in two sub-samples $W_{T_1}$ and $W_{T_2}$ of the same size $n$. Then the smoothing parameter is given by
		$$\overline{\epsilon_T}=\inf\left\{\epsilon>0:W_{T_2}\subset B(W_{T_1},\epsilon) \right\},$$  which is easy to calculate.  Indeed, for each point $W_i \in W_{T_2}  \ \  i=1, \ldots n$,  let $d_i$ stand for the distance of $W_i$ to its nearest neighbor in $W_{T_1}$  and let $\mathbf d=(d_1, \ldots, d_{n})$.  Then $\overline{\epsilon_T} = \max_{i=1,\ldots, n} d_i$. A more robust version is to take $\epsilon_T$ as the $1-\delta$  quantile of the vector $\mathbf d$ for a small value of $\delta$.  
					\end{itemize}
\end{remark}

\begin{remark} \label{choiceofr} 
The consistency of the $r$-convex hull estimator $C_r(W_T)$ of the trajectory $W_T$
is established in Theorem \ref{thm:rhullcons} under the sometimes unrealistic assumption 
that a lower bound is known for the maximal value of $r_0$ under which the set $S$
is $r_0$-convex. Here we suggest a way of choosing the parameter $r$ in a  data dependent manner.  If the set $S$ is $r$--convex for some $r>0$, then it is also $r'$--convex for any
$0<r'<r$ and a sufficiently small value of the parameter will do the job. However, in practice this does not provide a guide to choose the parameter. One may choose $r$ in a data dependent way by selecting $\hat{r}$ satisfying 
$$
d_H\left(C_{\hat{r}}(W_T), B(W_T, \epsilon_T)\right) \leq \inf_{r>0} d_H\left(C_r(W_T), B(W_T, \epsilon_T)\right)+\delta
$$
%\todo[color=red]{Respecto a la consulta de Gabor: Cambiamos la definici\'on pues en principio $C_r$ no es continua
%respecto de $d_H$ como funci\'on de $r$}
for an arbitrary small $\delta>0$, where $\epsilon_T$ is chosen as suggested by Theorem \ref{thm:rbmsausage}.
Thus, $\hat{r}$ is a value that makes the $r$--convex hull as close as possible to the RBM--sausage. It is an easy exercise to prove that $C_{\hat{r}}(W_T)$ is a consistent estimate of $S$
under the conditions of Theorem \ref{thm:rhullcons}, taking $\delta\rightarrow 0$.
\end{remark}

We close this section by establishing rates of convergence of the $r$-convex hull of the
trajectory $W_T$. For the simplicity of the exposure we concentrate on the 
$2$-dimensional case. The argument may easily be generalized for $d>2$
to obtain
$\mathbb{E}\Big(d_\mu(S_T,S)\Big)=\mathcal{O}\big(\log(T)/T)^{-2/(d+1)}\big).$

\begin{theorem}
Let $S\subset \mathbb{R}^2$ be a non-empty, connected and compact set such that $S=\overline{\textnormal{int}(S)}$. Suppose that a ball of radius $r$ rolls freely in $S$ and in $\overline{S^c}$. Let $\{X_t:t\ge 0\}$ be a reflected
Brownian motion in $\textnormal{int}(S)$.  Let $S_T=C_r(W_T)$.
Then
\[
\mathbb{E}\Big(d_\mu(S_T,S)\Big)=O((\log(T)/T)^{-2/3})~,
\]
where $\mu$ denotes the Lebesgue measure.
\end{theorem}

\begin{definition} Let $x\in \mathbb{R}$, $r>0$ and $\mathcal{E}_{x,r}=\big\{\mathcal{B}(y,r):y\in
\mathcal{B}(x,r)\big\}$. Following \cite{pateiro:13} the family of subsets
%\todo[color=red]{family of what? Precisar la definici\'on, por favor.}
$\mathcal{U}_{x,r}$ is said to be unavoidable for
$\mathcal{E}_{x,r}$, if, for all $\mathcal{B}(y,r)\in \mathcal{E}_{x,r}$
there exists $U\in \mathcal{U}_{x,r}$ such that $U\subset \mathcal{B}(y,r)$.
\end{definition}

\begin{proof} First observe that, by Proposition \ref{rollthennontrap}, the reflected Brownian motion exists and is non-trap. Let $S^{(r/2)}$ be the $r/2$-inner parallel set of $S$ defined in (\ref{inner}). 
For every $x\in S$ we may take an unavoidable family $\mathcal{U}_{x,r}$ with $6$ elements, such that, for all $U\in \mathcal{U}_{x,r}$, $\mu(U\cap S)\geq L_1r^2$ if $x\in S^{(r/2)}$, and $\mu(U\cap S)\geq L_2 r^{1/2}d(x,\partial S)^{3/2}$ if $x\in S\setminus S^{(r/2)}$   where $L_1$ and $L_2$ are positive constants (see Propositions 1 and 2 in \cite{pateiro:13}).
Then
\begin{align} \label{eq1}
\mathbb{E}\big(d_\mu(S,S_T)\big)&\leq  \int_{S}\sum_{U\in \mathcal{U}_{x,r}}\mathbb{P}\big(U\cap W_T=\emptyset\big)dx\nonumber\\
 &=\int_{S_1}\sum_{U\in \mathcal{U}_{x,r}}\mathbb{P}\big(U\cap W_T=\emptyset\big)dx+\int_{S_2}\sum_{U\in \mathcal{U}_{x,r}}\mathbb{P}\big(U\cap W_T=\emptyset\big)dx.
\end{align}
Where $S_1=S\setminus S^{(r/2)}$ and $S_2=S^{(r/2)}$. First we
bound the second term in (\ref{eq1}). Define $\delta=
L_2r^2/2$, $n$ as in (\ref{eqn}). Since a $r$-convex set is also
$r'$-convex, with $r'<r$, we can assume without loss of generality
that $r/2<1$ and $\beta/\delta>1$. Define
$t_i=\frac{i}{\alpha}\log \frac{\beta}{\delta}$ for
$i=1,\ldots,n$. For all $U\in \mathcal{U}_{x,r}$,
$$\PROB\{U\cap W_T=\emptyset\} \le  \PROB\{\forall i\in \{1,2,\ldots,n\}: X_{t_i}\notin U \}.$$
Now, if we proceed as in (\ref{eq2}),
$$\PROB\{\forall i\in \{1,2,\ldots,n\}: X_{t_i}\notin U \}= \PROB\{X_{t_n}\notin U| X_{t_{n-1}}\notin U \}
 \PROB\{\forall i\in \{1,2,\ldots,n-1\}: X_{t_i}\notin U \},$$

and by Proposition \ref{prop:burdzy} (iii),
\begin{eqnarray*}
\PROB\{X_{t_n}\notin U| X_{t_{n-1}}\notin U \}
& \le & 1-L_2r^2 +\beta e^{-\alpha (t_n-t_{n-1})}
\\
& = & 1-L_2r^2 + \delta  \\
& = & 1-\frac{L_2r^2}{2},
\end{eqnarray*}
by the definition of $\delta$.
Hence, we have
\[
\PROB\{\forall i\in \{1,2,\ldots,n\}: X_{t_i}\notin U \}
\le \left(1-\frac{L_1r^2}{2}\right)^n
\le \exp\left(- n \frac{L_1r^2}{2}\right)~.
\]
Then the second term is $\mathcal{O}(e^{-k_1T})$ for some constant
$k_1>0$. To deal with the first term, we proceed the same
way as before. If we take $\delta= \frac{L_2 r^{1/2}d(x,\partial
S)^{3/2}}{2}$ (observe that $\beta/\delta>1$), we have that
\begin{align*}
\PROB\{\forall i\in \{1,2,\ldots,n\}: X_{t_i}\notin U \}
\le& \left(1-\frac{L_2r^{1/2}d(x,\partial S)^{3/2}}{2}\right)^n\\
\le& \exp\left(- n \frac{L_2r^{1/2}d(x,\partial S)^{3/2}}{2}\right)\\
\le& \exp\left(- \frac{T}{\frac{1}{\alpha}\log\big(\frac{2\beta}{L_2r^{1/2}d(x,\partial S)^{3/2}}\big)} \frac{L_2r^{1/2}d(x,\partial S)^{3/2}}{2}\right)~.
\end{align*}
Then,
\begin{align*}
\int_{S_1}\sum_{U\in \mathcal{U}_{x,r}}\mathbb{P}\big(U\cap W_T=\emptyset\big)dx
\leq & \ 6\int_{S_1}\exp\left(- \frac{T}{\frac{1}{\alpha}\log\big(\frac{2\beta}{L_2r^{1/2}d(x,\partial S)^{3/2}}\big)} \frac{L_2r^{1/2}d(x,\partial S)^{3/2}}{2}\right)dx~.
\end{align*}
Let  $F(u)=\mu\{x\in S: d(x,\partial S)\leq u\}$ be the
distribution of the distance to the boundary $\partial S$ with
respect to the Lebesgue measure. Since the set $\overline{S^c}$
has positive reach, $F$ is polynomial in $u$) (see
\cite{pateiro:09}) and $m=\max_{u\in [0,r/2]}F'(u)<\infty$.  Then,
if $u=d(x,\partial S)$,
\begin{align} \label{eq3}
\int_{S_1}\sum_{U\in \mathcal{U}_{x,r}}\mathbb{P}\big(U\cap W_T=\emptyset\big)dx \leq & \ 6m\int_{0}^{r/2} \exp\left(- \frac{T}{\frac{1}{\alpha}\log\big(\frac{2\beta}{L_2r^{1/2}u^{3/2}}\big)} \frac{L_2r^{1/2}u^{3/2}}{2}\right)du~.
\end{align}
 Now let $z=\alpha TL_2r^{1/2}u^{3/2}/2\doteq c_1Tu^{3/2}$. 
Then 
%$(2/3)(Tc_1)^{-2/3}z^{-1/3}dz=du$, and
\begin{multline*}
\int_{0}^{r/2} \exp\left(- \frac{T}{\frac{1}{\alpha}\log\big(\frac{2\beta}{L_2r^{1/2}u^{3/2}}\big)} \frac{L_2r^{1/2}u^{3/2}}{2}\right)du=\\
\frac{2}{3}(Tc_1)^{-2/3}\int_{0}^{c_1T(r/2)^{3/2}} z^{-1/3}\exp\left(- \frac{z}{\log\big(\frac{\beta\alpha T}{z}\big)}\right)dz.\\
\end{multline*}
Since $r^{3/2}<1$, 
\begin{equation} \label{eq0} \frac{1}{T^{\frac{2}{3}}}\int_{0}^{c_1T(\frac{r}{2})^{\frac{3}{2}}} \frac{1}{z^{\frac{1}{3}}}\exp\left(- \frac{z}{\log\big(\frac{\beta\alpha T}{z}\big)}\right)dz\leq \frac{1}{T^{\frac{2}{3}}}\int_{0}^{c_1T}\frac{1}{z^{\frac{1}{3}}} \exp\left(- \frac{z}{\log\big(\frac{\beta\alpha T}{z}\big)}\right)dz.
\end{equation}
Taking $T$ large enough such that $c_1T>\alpha\beta$, we can majorize the right hand side of (\ref{eq0}) by  
$$
\frac{1}{T^\frac{2}{3}}\int_0^{\alpha\beta} \frac{1}{z^{\frac{1}{3}}} \exp\left(- \frac{z}{\log\big(\frac{\beta\alpha T}{z}\big)}\right)dz+
\frac{1}{T^\frac{2}{3}}\int_{\alpha\beta}^{c_1T} \frac{1}{z^{\frac{1}{3}}} \exp\left(- \frac{z}{\log(T)}\right)dz~.
$$
Finally, taking $s= \frac{z}{\log(T)}$, 
\begin{align*}
\int_{\alpha\beta}^{c_1T} \frac{1}{z^{\frac{1}{3}}} \exp\left(- \frac{z}{\log(T)}\right)dz=&\log(T)^{2/3}\int_{\alpha\beta/\log(T)}^{c_1T/\log(T)} \frac{1}{s^{\frac{1}{3}}} \exp(-s)ds\\
\leq & \log(T)^{2/3}\int_{0}^{+\infty} \frac{1}{s^{\frac{1}{3}}} \exp(-s)ds=K(\log(T))^{2/3}.
\end{align*}
\end{proof}

%GABOR multiples retoques

\section{Reflected diffusions} 
\label{sec:diffusion}

As mentioned in the introduction, even though the reflected Brownian
motion is a natural model, it may be too simplistic for some applications.
In particular, the fact that its stationary distribution is uniform on the domain
fails to capture some important aspects of animal movement. In order to 
address this problem, we suggest the more general, albeit less understood, model
of reflected diffusions.

A reflected diffusion 
in a connected and open domain $D \subset \mathbb{R}^d$ 
corresponds to the solution of the stochastic equation 
	\begin{equation} \label{sdeg}
	X_t=X_0+ U_t+\int_0^t\eta(X_s)dL_s,
	\end{equation}
	where
	$$U_t=\int_0^t\sigma(X_s)dB_s+\int_0^t b(X_s)ds,$$
	and $B_t$ is a $d$-dimensional Brownian motion, $\eta$ denotes the inward unit vector on the boundary $\partial D$, $L$ is a continuous nondecreasing process with $L_0=0$, and
	\[L_t=\int_0^t\mathbb{I}_{\{X_s\in\partial D\}}dL_s.\] 
\cite{saisho:87} proved that under the geometric conditions \eqref{cond1} and \eqref{cond2} there exists a unique strong solution to \eqref{sdeg} whenever $\sigma$ and $b$ are Lipschitz functions.

By taking $\sigma=1$ and $b=0$, one recovers reflected Brownian motions. 
By other choices of $\sigma$ and $b$ one obtains a rich class of stochastic processes
with possibly non-uniform stationary distribution.
A simple special case is the reflected Brownian motion with constant drift $\mu$, that corresponds to $U_t=B_t+ t\mu$ in \eqref{sdeg}. \cite{hw:87} proved that in this case, the stationary distribution has the form $C(\mu)\exp(\langle\gamma(\mu),x\rangle)$ for some $C(\mu)>0$ and $\gamma(\mu)\in \mathbb{R}^d$.

Assuming that the domain $D$ is non-trap for the reflected diffusion $X_t$,
one may generalize Theorem \ref{thcons}. Indeed, by observing that the proof
of the theorem relies on the non-trap condition, we have that 
if $S$ is a compact set satisfying $S=\overline{S}$ and $int(S)$ is connected,
for any reflected diffusion in $int(S)$ 
for which $int(S)$ is non-trap,
	\begin{itemize}
			\item[i)] $d_H(S_T,S)\rightarrow 0$ as $T\rightarrow +\infty$, under condition $W_T\subset S_T\subset B(S,\eps_T)$ for some sequence $\eps_T\rightarrow 0$ and any $S_T$ which contains the trajectory $W_T=\{X_t:0\leq t\leq T\}$.
			\item[ii)] In particular, $C_r(W_T)$ and $D_T$ are consistent estimates since they satisfy the condition $W_T\subset S_T\subset B(S,\eps_T)$.
			\end{itemize}
On the other hand, the non-trap condition is close to be necessary for consistency of $D_T$. In fact, is necessary for the complete convergence with respect to the Hausdorff distance. Indeed, if $int(S)$ is a trap domain, there exists a closed ball $B$ with radius $\delta>0$ and a sequence $x_n\in int(S)$  such that $\mathbb{E}(T^{x_n}_{B})\rightarrow +\infty$. Let $\epsilon<\delta/4$. If $d_H(S_T,S)$ converges to $0$ completely, then
		\begin{equation} \label{eq50}
		\sum_{n=1}^\infty P(d_H(S_n,S)>\epsilon)=:A<\infty.
		\end{equation}
		By choosing $x_0\in int(S)$ such that $\mathbb{E}(T^{x_0}_{B})>A+1$, we have $\sum_{n=1}^\infty P(T_B^{x_0})>A$. On the other hand,  for all $n$ we have $\{\omega: T_B^{x_0}>0\}\subset \{\omega: d_H(D_n,S)>\epsilon\},$ contradicting \eqref{eq50}.\\

Understanding the geometric conditions for the domain that imply the non-trap
property for general reflected diffusions remains an interesting research problem.

Also, to obtain rates of convergence for general reflected diffusions remains an open problem. In particular, one needs to extend Proposition \ref{prop:burdzy} in \cite{burdzy:06} to the case of general diffusions with non-uniform stationary distribution.

\section{Some comments on the implementation of the estimators}
\label{coments}
In practice, given a trajectory $\{X_t:0\leq t\leq T\}$ in the plane, we can approximate the $r$-convex hull of the trajectory and the RBM-sausage estimator, by adapting the implementation in the i.i.d. case. The computation of the $r$-convex hull estimator is based on the algorithms presented by \cite{ede:83}. The details of its implementation in R  (\cite{cran}) can be found in \cite{pateiro:10}. Regarding the RBM-sausage estimator, the implementation is based on the  computation of intersections of pairs of balls of radius $\epsilon_T$. We refer to \cite{ede:95} for efficient algorithms on the structure of a union of balls. The code for the computation of the $r$-convex hull of the trajectory and the RBM-sausage estimator is available in a new release of the R package \texttt{alphahull 2.0} that includes specific functions for these methods.

%Before starting with the technical aspects and, to motivate the interest in the problem, we would like to discuss the application of the proposed methodology in home range estimation from animal tracking data, through real data examples. }
%
%

Next, we present an illustration of the estimation of a non-convex compact set $S\subset\mathbb{R}^2$ from a simulated trajectory of a
planar reflected Brownian motion. An application to real data sets in the context of home-range estimation is discussed in Section \ref{sec:intro}.

Let $C$ be the region delimited by the so-called {\it{crooked egg curve}}. The polar equation of the crooked egg curve is 
\[r=\sin^3(\theta)+\cos^3(\theta).\]
%\todo[color=red]{Y el agujero dentro del huevo?}
Points $(x,y)$ in the curve satisfy 
$(x^2+y^2)^2-(x^3+y^3)=0$. Let us consider the set $S=C\setminus \mathring{\mathcal{B}}$ with $\mathring{\mathcal{B}}=\mathring{\mathcal{B}}\big((0.05,0.6),0.15\big)$. 
%We have considered here supports limited by a \textit{crooked egg curve} and a circumference.
In Figure \ref{fig:rbm2} we have simulated trajectories of a reflected Brownian motion on $S$. The reflection is pushed in the direction of the inward unit normal vector on the boundary. We represent in red, the boundary of the RBM-sausage $D_T$. We analyse the behaviour of $D_T$ with respect to how much the estimation differs when the algorithm is applied to a path from time 0 to $T$, or to the same path from time 0 to $2T$, etc. (short trajectories). The behaviour of the estimator is similar to that of the BBMM represented in green. The results based on the BBMM are provided in the R package BBMM, see \cite{nielson:13}. Since these are simulated data, the choice of the location error is based on visual exploration of the result. We assume that the trajectories are observed at a high sampling rate and the parameter corresponding to the time between successive observations is set small.
%\begin{figure}[!h]%
%\centering
%%\includegraphics[scale=0.4]{crookedN10000h0_001ini0_3.pdf}
%\includegraphics[scale=0.27]{crookedN10000h0_001ini0_3rhull.pdf}
%\includegraphics[scale=0.27]{crookedN10000h0_001ini0_3DW.pdf}
%\caption{In gray, support limited by a crooked egg curve and a circumference with %center $(0.05,0.6)$ and 
%radius $r=0.15$. In blue, path of a simulated reflected Brownian motion starting at $(0.3,0.3)$ with size $N=10000$ and step $h=0.001$. Left, in red, boundary of the $r$-convex hull estimator for {\mbox{$r=0.1$}}. Right, in red, boundary of the RBM-sausage $D_T$ for $\epsilon_T=0.025$.}
%\label{fig:rbm}%
%\end{figure}

\begin{figure}[!h]%
	\centering
	\includegraphics[scale=0.27]{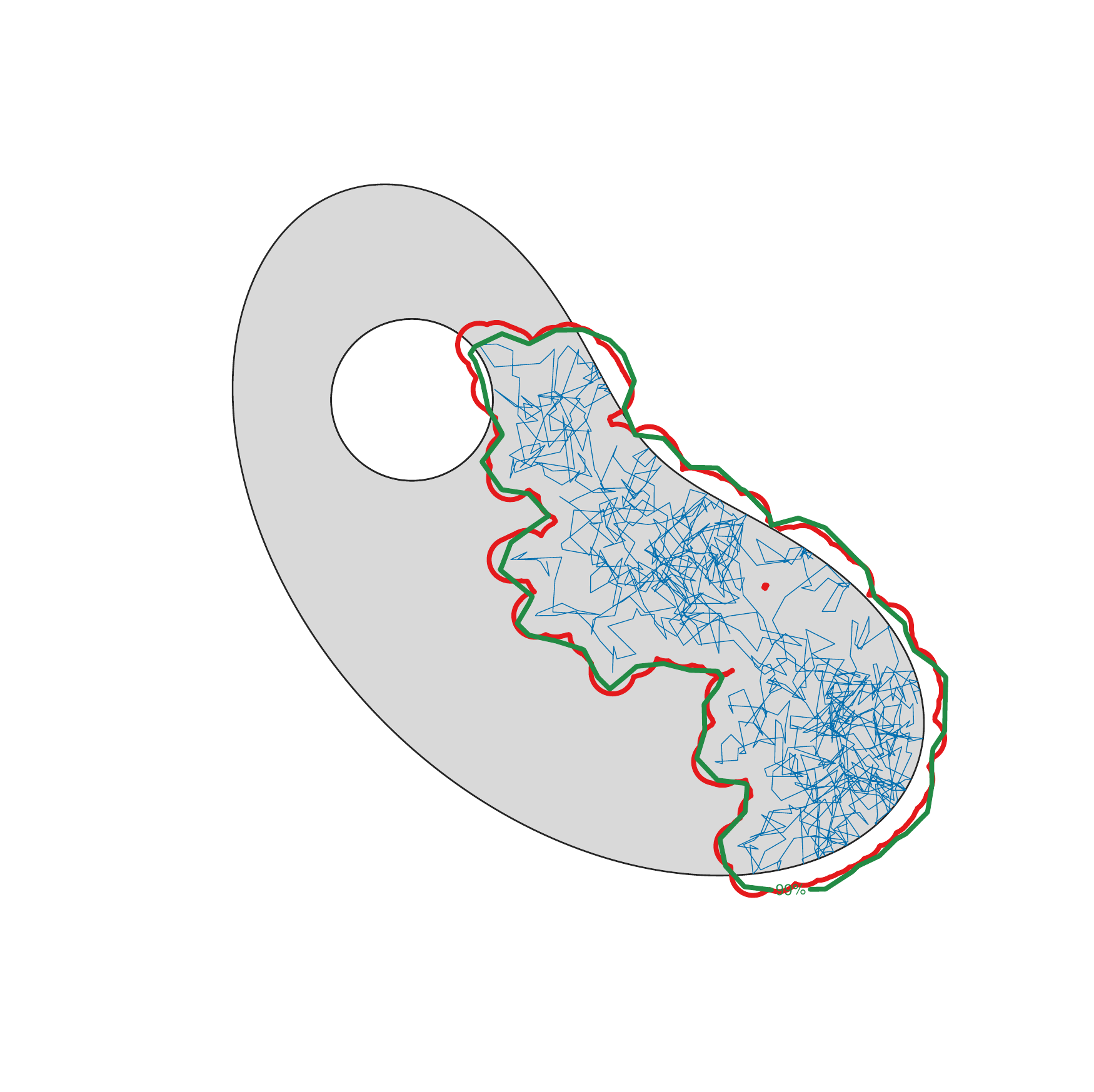}
	\includegraphics[scale=0.27]{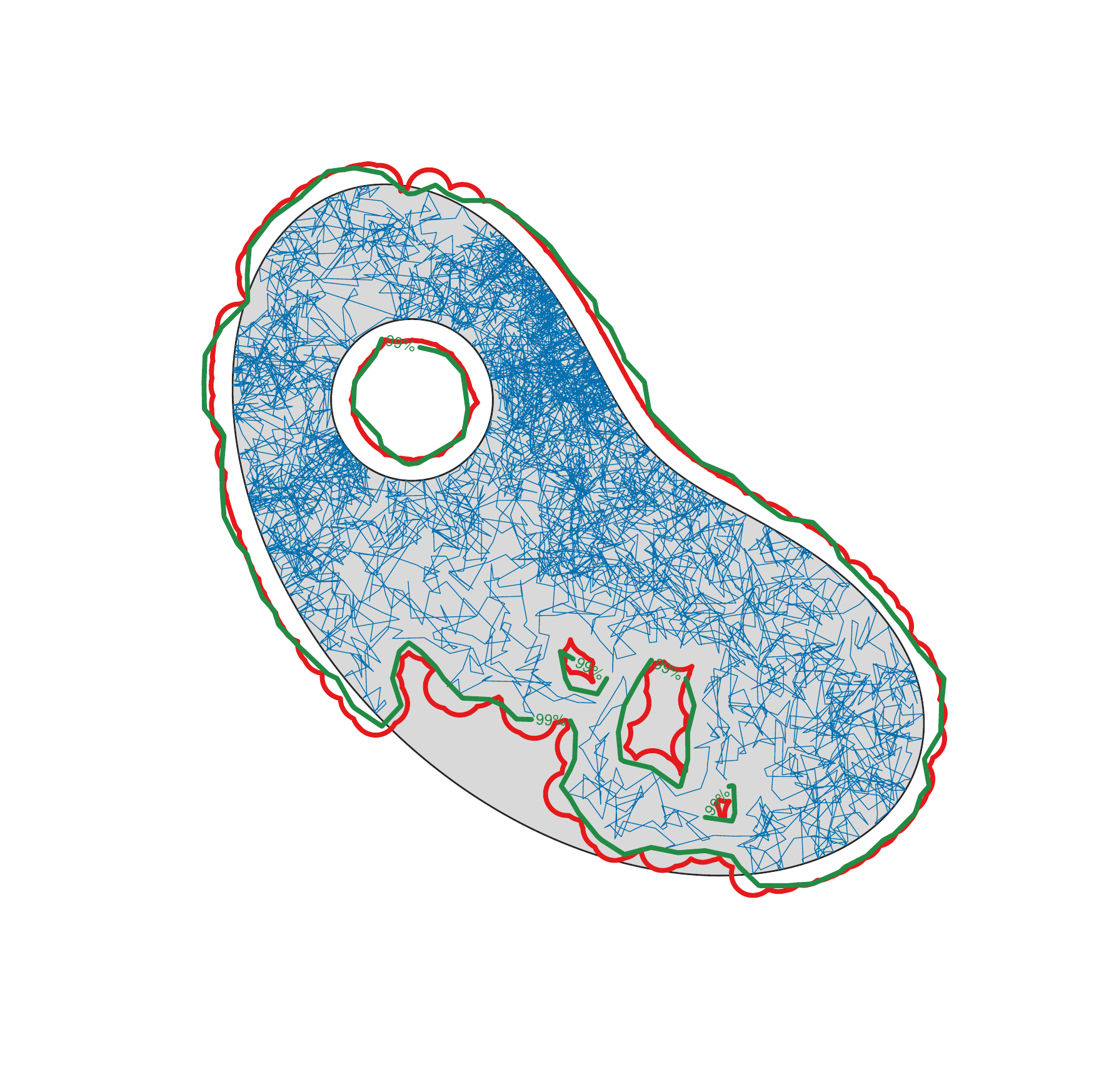}
	\includegraphics[scale=0.27]{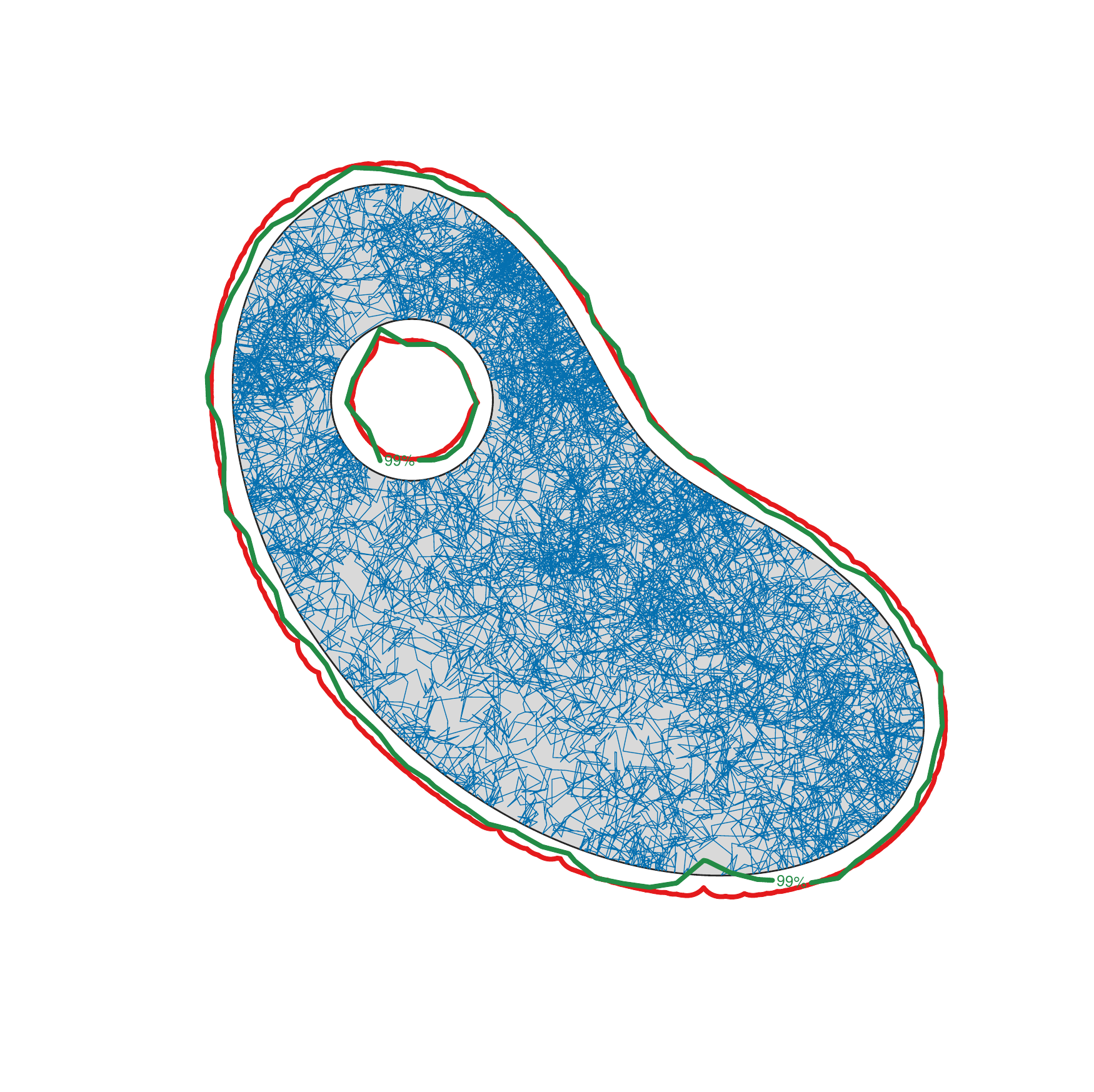}
	%\caption{In gray, support limited by a crooked egg curve and a circumference with %center $(0.05,0.6)$ and 
	%radius $r=0.15$. In blue, path of a simulated reflected Brownian motion starting at $(0.3,0.3)$.  Left, size $N=1000$ and step $h=0.001$. Right, size $N=5000$ and step $h=0.001$.}%
	\caption{In gray, support limited by a crooked egg curve and a circumference. % with %center $(0.05,0.6)$ and 
		%radius $r=0.15$. 
		In blue, path of a simulated reflected Brownian motion.  From left to right, size $N=1000$, $N=5000$ and $N=10000$ (with step $h=0.001$). In red, boundary of the RBM-sausage $D_T$ for $\epsilon_T=0.04$. In green, home-range estimation based on the BBMM (99\% utilization distribution estimated).}%
	\label{fig:rbm2}%
	
\end{figure}

In Figure \ref{fig:rbm3} we show the $r$-convex hull estimator and analyse its behaviour for lower sampling rates. We start with a simulated reflected Brownian motion with size $N=10000$ and step $h=0.001$ (left) and then select $N=2000$ (middle) and $N=500$ (right) equally spaced in time points. We also show the results based on the BBMM, choosing increasing values of the parameter corresponding to the time between successive observations.  As expected, see \cite{horne:07}, the  BBMM identifies the movement path with progressively less confidence as the time interval increases.

\begin{figure}[!h]%
	\centering
	\includegraphics[scale=0.27]{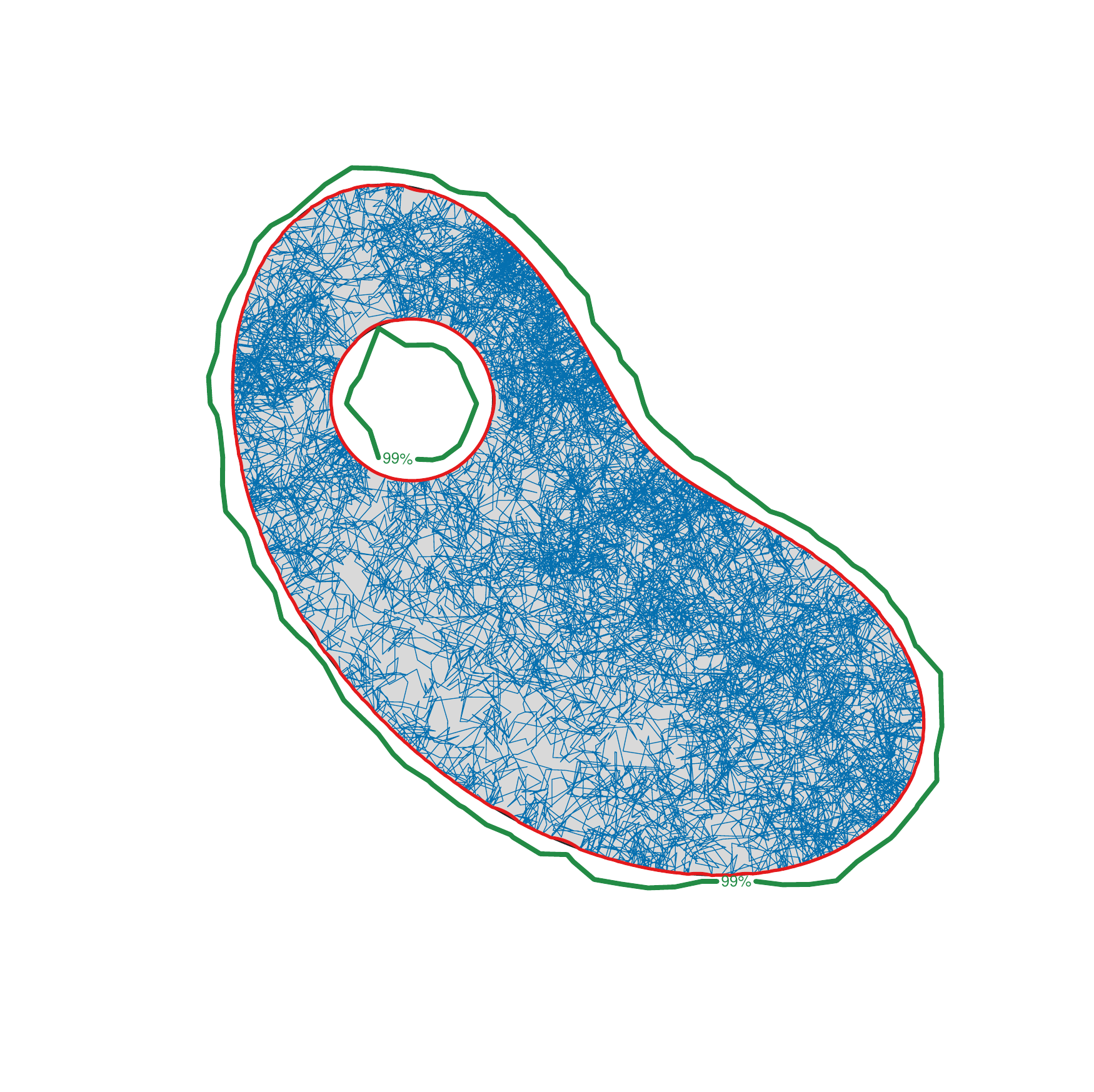}
	\includegraphics[scale=0.27]{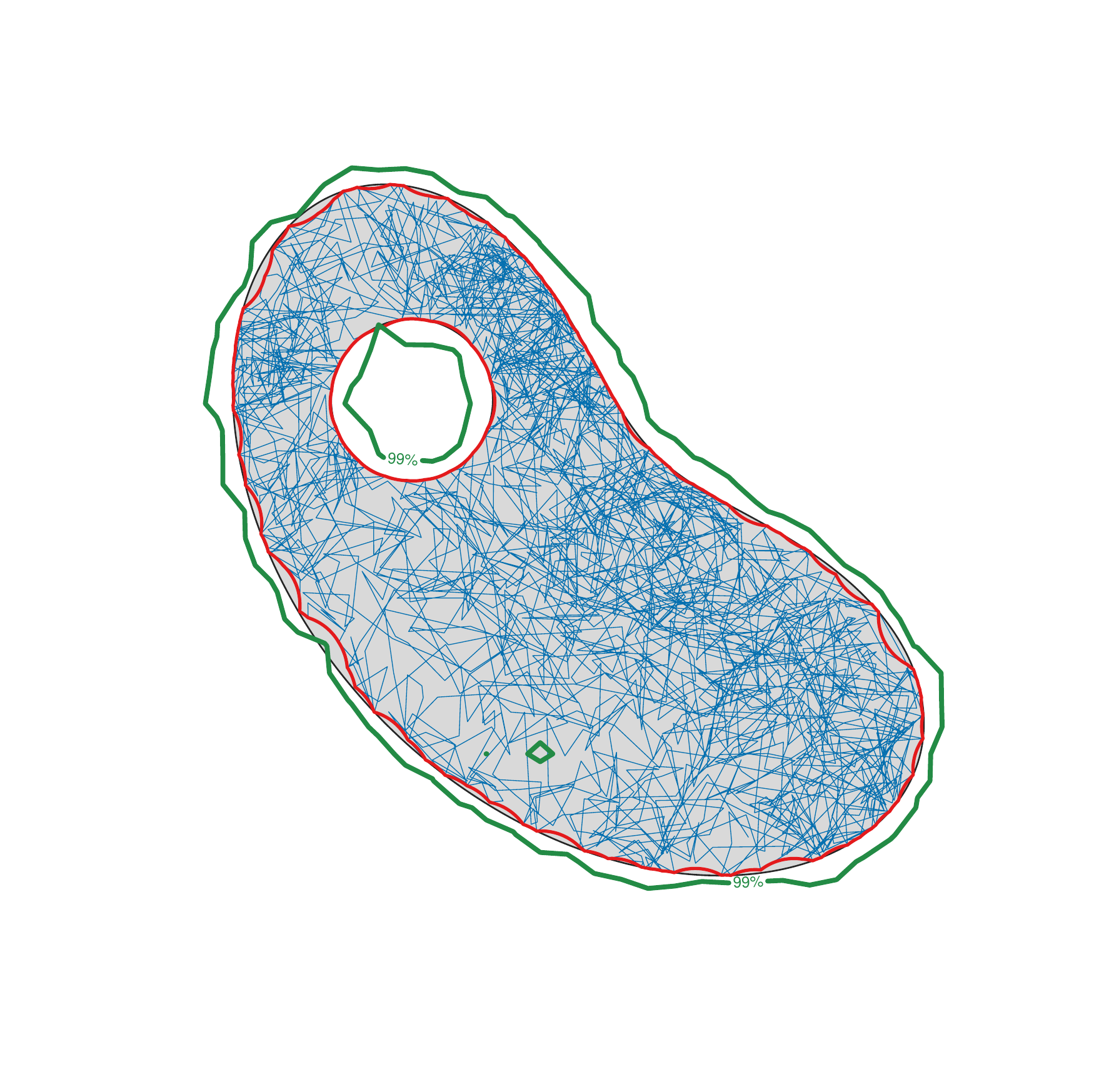}
	\includegraphics[scale=0.27]{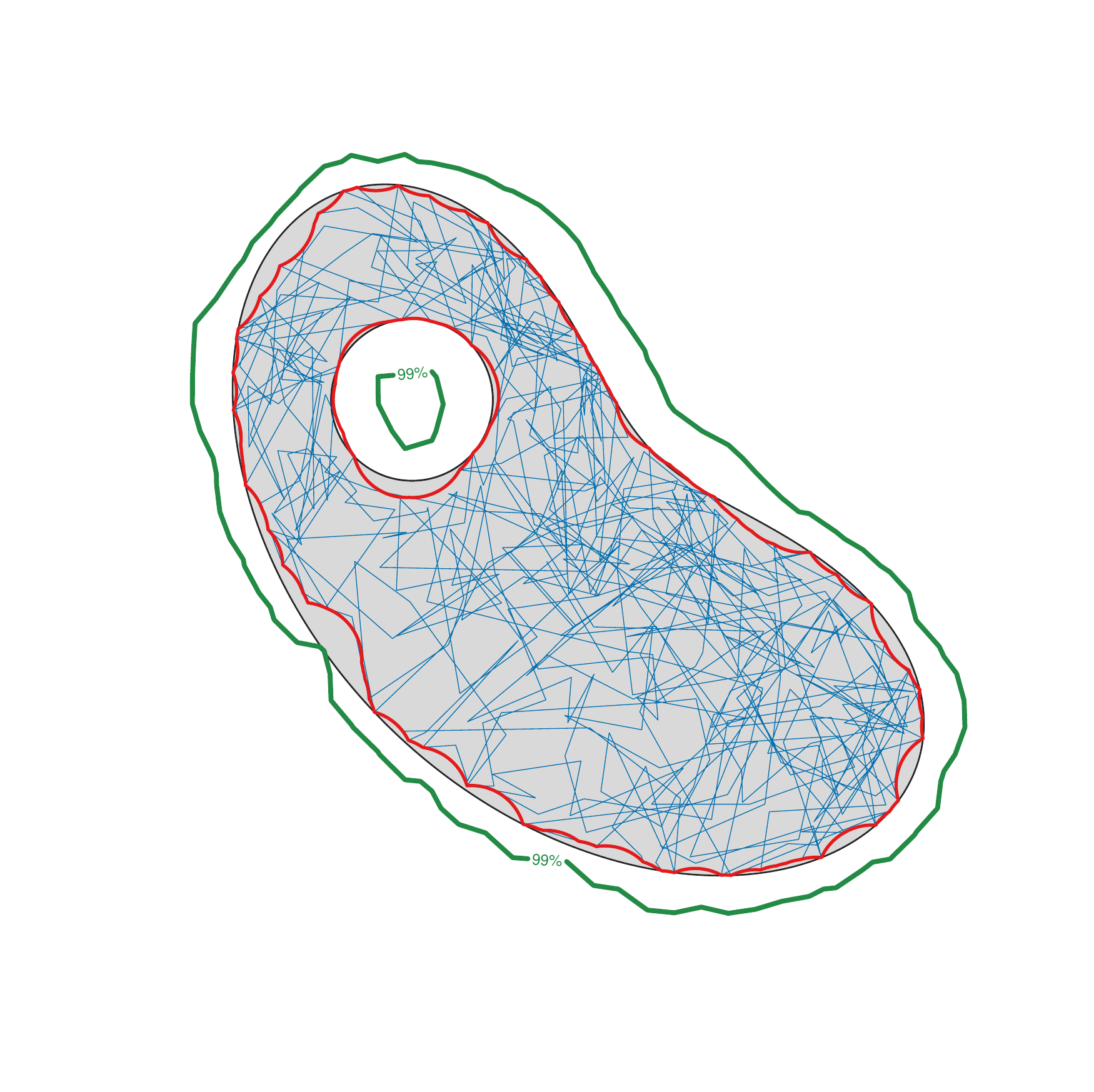}
	\caption{In blue, path of a simulated reflected Brownian motion with size $N=10000$ and step $h=0.001$. In red, boundary of the $r$-convex hull estimator for {\mbox{$r=0.1$}} for the original path (left),  $N=2000$ (middle) and $N=500$ (right) equally spaced in time points. In green, home-range estimation based on the BBMM (99\% utilization distribution estimated).}%
	\label{fig:rbm3}%
	
\end{figure}

\section{Concluding remarks}

 We present a new approach to estimate the support of the stationary distribution of a reflected diffusion based on a single trajectory. We consider two estimators: the $r$-convex hull of the trajectory and the Devroye-Wise estimate. Under the non-trap assumption both estimates are consistent with respect to the Hausdorff distance, as well as the distance in measure. These conditions are necessary in order to obtain complete consistency. We also provide algorithms for efficient computation of the estimates. 

Our main application is the home-range estimation problem. We illustrate the behaviour of our proposal through two real data examples as well as simulated data. We also compare with the BBMM method in two different setups: short trajectories and low sampling rates. The results are encouraging. 

We study in detail the case of the reflected Brownian motion and we obtain rates of convergence with respect to the Hausdorff metric, the distance in measure, and the $L^1$ distance. 

%GABOR eliminado
%\textcolor{magenta}{According to \cite{wal:15} third generation estimator methods for home-range problems are mainly MKDE and BBMM.
%We hope that reflected diffusion will open the gate to the fourth generation models. Our contribution is a first step in this direction.} 

\section*{Acknowledgements} 
The authors would like to thank Dr. Stephen Blake, of the Max Planck Institute for Ornithology, for facilitating access to the data sets discussed in Section \ref{sec:intro}. 
They are also grateful to the reviewers for their valuable comments and suggestions. 
The first and second authors acknowledge financial support from grant CSIC 605/34.
The third author was supported by the Spanish Ministry of Science and Technology grant MTM2012-37195.
The fourth author acknowledges financial support from the Spanish Ministry of Economy and Competitiveness
and ERDF funds (MTM2013-41383-P).
%\clearpage

\end{document}